\documentclass[submission,copyright,creativecommons]{eptcs}

\newcommand{\enc}[1]{\llbracket #1 \rrbracket}
\newcommand{\handler}{\mathit{H}}

\newcommand{\N}{{\cal N}}
\newcommand{\Var}{{\cal V}}
\newcommand{\Chn}{{\cal C}}
\newcommand{\NA}{a}
\newcommand{\NB}{b}
\newcommand{\NC}{c}
\newcommand{\ND}{d}
\newcommand{\NE}{e}
\newcommand{\NX}{x}
\newcommand{\NY}{y}
\newcommand{\NZ}{z}
\newcommand{\NK}{k}
\newcommand{\NL}{l}
\newcommand{\NM}{m}
\newcommand{\NN}{n}
\newcommand{\NT}{t}

\newcommand{\NGG}{g}
\newcommand{\fn}[1]{\mathsf{fn}(#1)}
\newcommand{\fnn}[1]{\mathsf{fnn}(#1)}
\newcommand{\bn}[1]{\mathsf{bn}(#1)}
\newcommand{\n}[1]{\mathsf{n}(#1)}
\newcommand{\fo}[1]{\mathsf{fo}(#1)}

\newcommand{\role}{\rho}

\newcommand{\varx}{x}

\newcommand{\PP}{P}
\newcommand{\PQ}{Q}
\newcommand{\PR}{R}
\newcommand{\PS}{S}

\newcommand{\inact}{0}
\newcommand{\parop}{\;|\;}
\newcommand{\rest}[1]{(\nu #1)}

\newcommand{\msg}{l}
\newcommand{\send}[4]{#1!#4}

\newcommand{\receive}[4]{#1?#4}

\newcommand{\rep}[1]{! #1}
\newcommand{\match}[4]{[ #1 = #4 ]}

\newcommand{\red}{\rightarrow}

\newcommand{\lts}[1]{\xrightarrow{#1}}
\newcommand{\ltss}[2]{\xrightarrow#1{}\negthickspace^#2}
\newcommand{\subst}[2]{\{#1/#2\}}

\newtheorem{proposition}{Proposition}
\newtheorem{lemma}{Lemma}
\newtheorem{theorem}{Theorem}
\newtheorem{corollary}{Corollary}
\newtheorem{definition}{Definition}

\newenvironment{proof}[1][Proof.]{\begin{trivlist}
\item[\hskip \labelsep {\bfseries #1}]}{\end{trivlist}}

\newcommand{\rulename}[1]{\text{\small \textsc{#1}}}

\newenvironment{theorema}[2]{\begin{trivlist}
\item[\hskip \labelsep {\bfseries Theorem #1}] {#2}}{\end{trivlist}}

\newenvironment{lemmata}[2]{\begin{trivlist}
\item[\hskip \labelsep {\bfseries Lemma #1}] {#2}}{\end{trivlist}}

\newenvironment{prop}[2]{\begin{trivlist}
\item[\hskip \labelsep {\bfseries Proposition #1}] {#2}}{\end{trivlist}}

\usepackage{amssymb}
\usepackage{makeidx}  
\usepackage{amsmath}
\usepackage{mathpartir}
\usepackage{stmaryrd}
\usepackage{xcolor}
\providecommand{\event}{ICE 2019} 
\usepackage{breakurl}             
\usepackage{underscore}           

\title{The $C_\pi$-calculus: a Model for Confidential Name Passing
}
\author{
Ivan~Proki\'c
\institute{Faculty of Technical Sciences, University of Novi Sad, Serbia}
}
\def\titlerunning{The $C_\pi$-calculus: a Model for Confidential Name Passing}
\def\authorrunning{Ivan Proki\' c}
\begin{document}
\maketitle

\begin{abstract}
Sharing 
confidential information 
in distributed systems is a necessity in many applications, however, it
opens the problem of controlling 
information sharing even among trusted parties.
In this paper, we present a  formal model in which dissemination 
of information is disabled at the level of the syntax in a  
direct way. 
We introduce a subcalculus of the $\pi$-calculus 
%
in which channels are considered
as confidential information. 
The 
only difference 
with respect to the $\pi$-calculus is that channels 
once received cannot be forwarded  later on. 
By means of examples, we give an initial idea of how some privacy
notions already studied in the past, such as group creation 
and  name hiding, 
can be represented without any additional language constructs.
We also present an encoding of the (sum-free) $\pi$-calculus 
in our calculus. 
\end{abstract}

\section{Introduction}

Sharing sensitive information over the internet has become an everyday routine: 
sending personal data and/or credit card number for online shopping is just one of 
the examples where the sensitive information can be 
disposed to other 
parties. 
Such cases open the problem of controlling 
information sharing even among trusted parties. 
The problem of privacy can (and must) be perceived both from a legal and technological point of view. 
One of the first who explored privacy in the information age, a legal scholar, 
Alan Westin had recognized that 
``Building privacy controls into emerging technologies 
will require strong effort..."~\cite{westin2003social}.
On the other hand, new technologies can also provide new ways to deal 
with 
privacy problems~\cite{DBLP:conf/fm/TschantzW09}. 
According to Solove~\cite{solove2005taxonomy}, there are four types of privacy violation: 
invasions, information collection, information processing, and information 
dissemination (see also \cite{DBLP:journals/lmcs/KouzapasP17}). 
The focus of this paper will be on presenting the techniques for controlling 
information dissemination in distributed systems. 

Although there is a further taxonomy for information dissemination violation by Solove, 
all these sub-types roughly speak about harms of revealing the personal data or 
threats of spreading information.
In distributed systems where communication of 
entities is central, 
controlling the flow of confidential information poses some obstacles. 
The capability of forwarding, that makes it possible to disseminate received information, 
may be recognized as one 
problem in controlling such systems.

Even in the examples of well-structured communications between two 
parties, such as the ones respecting the protocols specified by session types~\cite{DBLP:conf/esop/HondaVK98}, 
it can be permissible for any party to forward (i.e., delegate) its session end-point. 
The session delegation is crucial for establishing sessions between the two parties~\cite{DBLP:journals/acta/GayH05, DBLP:journals/iandc/Vasconcelos12}, such as in 
\[
\rest{\mathit{session}}\send{channel}\role\msg{\mathit{session}}.\mathit{Alice} \quad\parop\quad \receive{\mathit{channel}}\role\mgs{\NX}.\mathit{Bob}
\]
where $\mathit{Alice}$ creates a fresh channel $\mathit{session}$ and sends one end-point along $\mathit{channel}$ to $\mathit{Bob}$. 
However, it is also possible that the receiving party forwards the channel (cf. session delegation), as we may specify 
$\mathit{Bob}=\send{\mathit{forward}}\role\msg\NX.\mathit{Bob}'$. 
Such forwarding capability may be appealing to have in some cases (e.g., forwarding tasks from a master to a slave process), but considering $\mathit{session}$ is a channel created by $\mathit{Alice}$, 
pointing to some private data 
and shared exclusively with $\mathit{Bob}$, 
one might argue that $\mathit{Bob}$ should not gain the capability of session delegation 
just by receiving $\mathit{session}$. Hence, if we consider the name of channel $\mathit{session}$ to be confidential, $\mathit{Alice}$ should be the one who decides whether to let a third party knows about the channel. 
In addition, we may argue that in some cases there is no predefined set of parties that 
may receive $\mathit{session}$ in any future. 
Indeed, $\mathit{Alice}$ should be able to send $\mathit{session}$ end-point to any other party she decides.
Therefore, we may conclude that confidential information can also be shared in open-ended systems, 
where the set of users of the information cannot be statically predefined.

In this paper, we present a  formal model in which dissemination 
of information is disabled at the level of the syntax in a direct way. 
We build on the $\pi$-calculus~\cite{pi-calculus}, a process model tailored for 
communication-centric systems, by introducing a subcalculus which we call \emph{Confidential $\pi$-calculus},
 abbreviated $C_\pi$. 
The only information shared in our calculus are names of channels, so channels are the confidential information. 
This is the only difference of our model with respect to the $\pi$-calculus, names of channels are 
confidential and hence once received cannot be forwarded  later on. 
By means of examples, the paper gives an initial idea of how some privacy
notions already studied in the past, such as group creation 
and  name hiding, 
can be represented without any additional language constructs. 
We also define the non-forwarding property and show that, naturally, all $C_\pi$ processes satisfy this property. 
This result is then reused to differentiate the $\pi$-calculus processes that never forward received channels: if a $\pi$ process is bisimilar to a $C_\pi$ process then the $\pi$ process satisfies the non-forwarding property. 
We also propose an encoding from $\pi$-calculus into $C_\pi$-calculus and show its completeness.
The paper presents initial results of the investigation of the model, 
formalization of some of the results are left for an extended version of the paper.

The paper is organized as follows. In Section~\ref{sec:process-model}, we start by presenting the syntax 
and semantics of $C_\pi$-calculus, and we state some properties of the labeled transition system.  In Section~\ref{sec:behavioral} we define 
a behavioral equivalence relation, called strong bisimilarity. 
Using the definition of strong bisimilarity we state and prove that the closed domains for channels are directly representable in $C_\pi$. 
Another consequence of this result is the possibility of creating channels with similar behavior to $CCS$ channels~\cite{DBLP:books/sp/Milner80}. 
Using the strong bisimilarity relation and non-forwarding of $C_\pi$ processes, we also propose a method for differentiating $\pi$-calculus processes that never forward received channels. In addition 
Section~\ref{sec:examples} presents some further informal insights on $C_\pi$ and several 
interesting scenarios which are naturally represented in our model. 
Even though non-forwarding property restricts the syntax of the $\pi$-calculus, in Section~\ref{sec:encoding} we 
show the $C_\pi$ is expressive enough to model the $\pi$-calculus.
The base idea of the encoding is to create dedicated processes for each channel that 
handle sending the respective channels. 
In Section~\ref{sec:related work} we conclude and point to the related work.







\section{Process Model}\label{sec:process-model}
\begin{table}[t]
\[
\displaystyle
\begin{array}[t]{@{}rcl@{\;\;}l@{}}
 \pi & ::= & \quad \send\NA\role\msg\NK \quad \parop \quad \receive\NA\role\msg{x} \quad \parop \quad \match\NA\role\msg\NB\pi\\
  \PP & ::= & \quad  \inact \quad \parop \quad \pi.\PP  \quad \parop \quad \PP\parop\PP \quad \parop \quad \rest\NK \PP  \quad \parop \quad \rep\PP \\
\end{array}
\]
\caption{\label{tab:syntax} Syntax of prefixes and processes.}
\end{table}
In this section, we present the syntax and semantics of $C_\pi$. 
The main difference with respect to the $\pi$-calculus processes is that in $C_\pi$ names received 
in an 
input cannot be later used as an object of an 
output, hence disallowing forwarding. 
Apart from this difference, the remainder of this section should come 
as no surprise to a reader familiar with the $\pi$-calculus.
To make a clear distinction between names of variables bound in input and names of channels, we introduce 
two disjoint countable sets $\Var$ and $\Chn$, where $\Var$ is the set of variables, ranged over by 
$\NX, \NY, \NZ, \ldots$, and $\Chn$ is the set of channel names, ranged over by $\NK, \NL, \NM, \ldots$. 
We denote with $\N$ the union of sets $\Var$ and $\Chn$, and we let $\NA, \NB, \NC, \ldots$ range over $\N$.


\paragraph{Syntax.} 
Table~\ref{tab:syntax} presents the syntax of the language.
An inactive process is represented with $\inact$.  
The prefixed process $\pi.\PP$ comprehends process $\send\NA\role\msg\NK.\PP$, which on name $\NA$ 
sends channel $\NK$ and then proceeds as $\PP$, 
process $\receive\NA\role\msg\NX.\PP$ which on name $\NA$ receives a channel and substitutes the received channel for $\NX$ in $\PP$, and the last prefix $\match\NA\role\msg\NB\pi.\PP$ which exhibits $\pi.\PP$ only if $\NA$ and $\NB$ are the same name.
Notice that  in our language, differently from the $\pi$-calculus, there is a syntactic 
restriction of the objects in prefixes: only a channel ($\NK$) can be sent and only a variable ($\NX$) can be used as 
a placeholder for a channel to be received. For example, $\pi$-calculus process $\receive\NA\role\msg\NX.\send\NB\role\msg\NX.\inact$ is not part of the $C_\pi$ syntax.
There is no restriction on subjects of prefixes and names to be matched, these can be either variables or channels ($\NA$).
Parallel composition $\PP \parop \PP$ stands for two processes simultaneously active, that may interact. 
Channel restriction $\rest\NK \PP$ expresses that a new channel $\NK$, known only to process $\PP$, is created. 
Replicated process $\rep\PP$ introduces an infinite behavior. 
Intuitively, consider $\rep\PP$ stands for an infinite parallel composition of copies of process $\PP$ (i.e., $\PP\parop\PP\parop \cdots$). Here, we do not consider the choice operator (i.e., the sum) since we believe it is not fundamental to our approach, but we remark choice can be added in expected lines.

In $\rest\NK \PP$ and $\receive\NA\role\msg\NX.\PP$, the channel $\NK$ and the variable $\NX$ are binding 
with scope $\PP$. 
The set of bound names $\bn\PP$, for any process $\PP$, is defined as the union of bound channels and bound variables in $\PP$. 
The set of free names $\fn\PP$ and the set of names $\n\PP$, for any $\PP$, are defined analogously.
In addition, we use $\fo\PP$ to denote the set of all free channels appearing as 
objects of output prefixes in process $\PP$.

\paragraph{Semantics.}

\begin{table}[t]
\[
\begin{array}[t]{@{}c@{}}
\inferrule[(out)]{}
{\send\NK\role\msg\NL.\PP\lts{\send\NK\role\msg\NL}\PP}
\qquad
\inferrule[(in)]{}
{\receive\NK\role\msg\varx.\PP\lts{\receive\NK\role\msg\NL}\PP\subst\NL\varx}
\qquad
\inferrule[(match)]
{\pi.\PP\lts{\alpha}\PP'}
{\match\NA\role\msg\NA\pi.\PP\lts{\alpha}\PP'}
\qquad
\inferrule[(res)]
{\PP\lts{\alpha}\PP' \quad k\notin\n\alpha}
{\rest\NK \PP \lts{\alpha}\rest\NK\PP'}
\qquad
\inferrule[(open)]
{\PP\lts{\send\NK\role\msg\NL}\PQ \quad k\not= l}
{\rest\NL\PP\lts{\rest\NL\send\NK\role\msg\NL}\PQ }
\vspace{2ex}\\
\inferrule[(par-l)]
{\PP\lts{\alpha}\PQ \quad \bn\alpha\cap\fn\PR=\emptyset}
{\PP\parop\PR\lts{\alpha}\PQ\parop\PR}
\qquad
\inferrule[(comm-l)]
{\PP\lts{\send\NK\role\msg\NL}\PP' \quad \PQ\lts{\receive\NK\role\msg\NL}\PQ'}
{\PP\parop\PQ\lts{\tau}\PP'\parop\PQ'}
\qquad
\inferrule[(close-l)]
{\PP\lts{\rest\NL\send\NK\role\msg\NL}\PP' \quad \PQ\lts{\receive\NK\role\msg\NL}\PQ' \quad \NL\notin \fn{\PQ}}
{\PP\parop\PQ\lts{\tau}(\nu \NL)(\PP'\parop\PQ')}
\vspace{2ex}\\
\inferrule[(rep-act)]
{\PP\lts{\alpha}\PP'}
{\rep\PP\lts{\alpha}\PP'\parop \rep\PP}
\qquad
\inferrule[(rep-comm)]
{\PP\lts{\send\NK\role\msg\NL}\PP' \quad \PP\lts{\receive\NK\role\msg\NL} \PP''}
{\rep\PP\lts{\tau} (\PP' \parop \PP'') \parop \rep\PP}
\qquad
\inferrule[(rep-close)]
{\PP\lts{\rest\NL\send\NK\role\msg\NL}\PP' \quad \PP\lts{\receive\NK\role\msg\NL} \PP'' \quad \NL\notin\fn\PP}
{\rep\PP\lts{\tau} \rest\NL(\PP' \parop \PP'')\parop  \rep\PP}
\end{array}
\]
\caption{\label{tab:Transition}LTS rules.}
\end{table}

We present an operational semantics for our model in terms of the labeled transition system, which build on observable labeled actions $\alpha$, defined as 
\[
\alpha ::= \quad \send\NK\role\msg\NL \quad \parop \quad \receive\NK\role\msg\NL \quad \parop \quad \rest\NL\send\NK\role\msg\NL \quad \parop \quad \tau
\]
Action $\send\NK\role\msg\NL$ sends the channel $\NL$ on the channel $\NK$, 
while $\receive\NK\role\msg\NL$ receives the channel $\NL$ on channel $\NK$. 
In action $\rest\NL\send\NK\role\msg\NL$ the sent channel $\NL$ is bound, 
and $\tau$ stands for internal action. 
Notice that, as in the $\pi$-calculus, names bound in the input (variables) cannot appear in 
labels of observable actions. To retain the same notation as for processes, we denote by 
$\fn\alpha$, $\bn\alpha$ and $\n\alpha$, the sets of free, bound and all names of observable $\alpha$, respectively. As we noted above, these sets contain only channels, and not variables.

The transition relation is defined inductively by the rules given in Table~\ref{tab:Transition}. 
Notice that the action labels and transition rules are defined exactly as in the $\pi$-calculus~\cite{pi-calculus}.
The symmetric rules for \rulename{(par-l)}, \rulename{(comm-l)} and \rulename{(close-l)} are elided from the table.
Rules \rulename{(out)}, \rulename{(in)} and \rulename{(match)} are consistent with the explanations 
of the corresponding syntactic constructs. 
Rule \rulename{(res)} ensures that the action of the process is the action of the process scoped over by channel restriction if the channel specified in restriction is not mentioned in the action.
Rule \rulename{(open)} opens the scope of the restricted channel, enabling the extrusion of its scope 
while ensuring that subject and the object of the action are different channels.
Rule \rulename{(par-l)} lifts the action of one of the branches, and the side condition ensures that 
the channel bound in the action is not specified as free in the other branch.
In rule \rulename{(comm-l)} two processes performing dual actions, one sending and other receiving 
$\NL$ along $\NK$, synchronize their actions in the respective parallel composition.
In rule \rulename{(close-l)} the channel sent ($\NL$) by the left process is bound and after the 
synchronization with the right process (performing the dual action), the scope of $\NL$ is closed
while avoiding unintended name capture.
Rules \rulename{(rep-act)}, \rulename{(rep-comm)} and \rulename{(rep-close)} describe the actions of a replicated process.
The first rule lifts the action of a single copy of the replicated process and activates $\rep\PP$ in parallel.
The second and the third rules show cases when two copies of replicated process synchronize their actions, 
either through communicating a free or bound channel, where, again, in both cases a copy of $\rep\PP$ is 
activated in parallel.
As usual, we identify $\alpha$-convertible processes, and thus, we use our transition rules up to $\alpha$-conversion when needed.

We now present some specific results of the transition relation in the $C_\pi$-calculus. 
Our first result shows the relation between the set of free channels appearing as objects of output prefixes in the process and the process redexes: this set 
can be (possibly) enlarged only by opening the scope of a bound channel. Even more, the input actions do not affect the set of free channels appearing as objects of output prefixes in the process.

\begin{lemma}\label{lemm:fo-and-transitions}
Let $\PP\lts{\alpha}\PP'$. 
\begin{enumerate}
\item If $\alpha=\receive\NK\role\msg\NL$ then $\fo{\PP'}=\fo{\PP}$.
\item If $\alpha=\send\NK\role\msg\NL$ then $\NL\in\fo\PP$ and $\fo{\PP'}\subseteq\fo{\PP}$.
\item If $\alpha=\rest\NL\send\NK\role\msg\NL$ then $\NL\in\bn\PP$ and $\fo{\PP'}\subseteq\fo{\PP}\cup\{\NL\}$.
\item If $\alpha=\tau$ then $\fo{\PP'}\subseteq\fo{\PP}$.
\end{enumerate}
\end{lemma}
\begin{proof}
The proof is by induction on the derivation $\PP\lts{\alpha}\PP'$. 
We only discuss the base case of $3.$, when rule \rulename{(open)} is applied. 
Then, $\PP=\rest\NL\PP_1$ and $\rest\NL\PP_1\lts{\rest\NL\send\NK\role\msg\NL}\PP'$ is derived from 
$\PP_1\lts{\send\NK\role\msg\NL}\PP'$. 
By $2.$ of this Lemma, we get $\NL\in\fo{\PP_1}$ and $\fo{\PP'}\subseteq\fo{\PP_1}$. 
Since $\NL\in\bn{\rest\NL\PP_1}$, we conclude $\fo{\PP'}\subseteq\fo{\rest\NL\PP_1}\cup\{\NL\}$.
\end{proof}

As a direct consequence of Lemma~\ref{lemm:fo-and-transitions} we get the next corollary.

\begin{corollary}\label{cor:fo-and-in-out}
\begin{enumerate}
\item If $\PP\lts{\receive\NK\role\msg\NL}\PP'$ and $\NL\notin\fo\PP$ then $\NL\notin\fo{\PP'}$.
\item If $\NL\notin\fo{\PP}$ then there is no process $\PP'$ and channel $\NK$ such that $\PP\lts{\send\NK\role\msg\NL}\PP'$.
\end{enumerate}
\end{corollary}


What we can conclude from Corollary~\ref{cor:fo-and-in-out} combining its two statements is that a $C_\pi$ process cannot send a channel that it previously has received if the channel was not specified as an object of an output prefix in the first place. To show that this property is preserved also by all possible redexes of the process let us first relate the set of free channels appearing in output prefixes of the set and any its execution trace. The result is a direct consequence of Lemma~\ref{lemm:fo-and-transitions}.

\begin{corollary}\label{cor:trace-fo}
If $\PP\lts{\alpha_1}\PP_1\lts{\alpha_2}\ldots\lts{\alpha_m}\PP_m$ then $\fo{\PP_m}\subseteq\fo{\PP}\cup\bn{\alpha_1}\cup\ldots\cup\bn{\alpha_m}$.
\end{corollary}

The next theorem states that if the channel received by a process was not previously specified as an object of an output prefix, it will not be sent in any of the process possible evolutions. Hence, for the $C_\pi$ processes 
forwarding a channel, in a sense that a process can send a channel he learns through receiving, is not possible. Before the theorem, we present the precise definition of non-forwarding. 

\begin{definition}[
Non-forwarding Property]\label{def:non-leaking-pi-processes}
A process $\PP_1$ satisfies the non-forwarding 
property if whenever 
\[
\PP_1\lts{\alpha_1}\PP_2\lts{\alpha_2}\ldots\lts{\alpha_{m}}\PP_{m+1}.
\]
then if $\NL\notin\fn{\PP_i}$ and $\alpha_i=\receive\NK\role\msg\NL$, for some $i=1, \ldots, m-1$, then $\alpha_j\not=\send{\NK'}\role\msg\NL$, for all $j=i+1, \ldots,m$.
\end{definition}

Notice that in the last definition we could also add the condition $\alpha_i\not=\rest\NL\send{\NK'}\role\msg\NL$, for all $j=i+1, \ldots,m$. But, since without loss of generality we can assume all bound outputs are fresh, we can omit such condition.
The next theorem attests that all $C_\pi$ processes respect the non-forwarding 
property, but in a more rigorous way, where the only restriction for the channel is not to be specified as the free object of any output prefix. 

\begin{theorem}[The $C_\pi$ Processes Respect the Non-Forwarding 
Property]\label{the:non-forwarding}
Let
\[
\PP_1\lts{\alpha_1}\PP_2\lts{\alpha_2}\ldots\lts{\alpha_{m}}\PP_{m+1}.
\]
Then, if $\NL\notin\fo{\PP_i}$ and $\alpha_i=\receive\NK\role\msg\NL$, for some $i=1, \ldots, m-1$, then $\alpha_j\not=\send{\NK'}\role\msg\NL$, for all $j=i+1, \ldots,m$.
\end{theorem}

\begin{proof}
Since without loss of generality we can assume all bound outputs are fresh
 and $\NL\notin\fo{\PP_i}$, using Corollary~\ref{cor:trace-fo} we get $\NL\notin\fo{\PP_j}$, for $j=i+1, \ldots,m+1$. Hence, by Corollary~\ref{cor:fo-and-in-out}~$2.$ we get $\alpha_j\not=\send{\NK'}\role\msg\NL$, for $j=i+1, \ldots,m$.
\end{proof}

The result of Theorem~\ref{the:non-forwarding} should come as no surprise, $C_\pi$ processes are designed to respect the non-forwarding property. However, considering the $\pi$-calculus, it appears to be nontrivial to differentiate processes that respect the non-forwarding property. To attack this goal, we will see in the next section that Theorem~\ref{the:non-forwarding} can be reused in Proposition~\ref{prop:non-forwarding-of-pi-processes}. 

\section{Behavioral equivalence}\label{sec:behavioral}
Based on the notion of observable actions, introduced in Section~\ref{sec:process-model}, we investigate some specific behavioral 
identities of our model. 
To this end, we introduce a behavioral equivalence, called strong bisimulation, which, colloquially speaking, relates two processes if one can play a symmetric game over them: each action of one process can be mimicked 
by the other (and with the order reversed), leading to two processes that are again related. 
The relation that we are interested in is the largest such relation, called strong bisimilarity.

\begin{definition}[Strong bisimilarity]\label{def:strong-bisimilarity}
The largest symmetric binary relation over processes $\sim$, satisfying 
\[
\mbox{if}\;\; \PP\sim\PQ \;\;\mbox{and}\;\; \PP\lts{\alpha}\PP', \;\;\mbox{where}\;\; \bn{\alpha}\cap\fn{\PQ}=\emptyset, \;\;\mbox{then}\;\; \PQ\lts{\alpha}\PQ' \;\;\mbox{and}\;\; \PP'\sim\PQ',
\]
for some process $\PQ'$, is called strong bisimilarity.
\end{definition}

Notice that, since our transition rules match those of the $\pi$-calculus, our strong bisimilarity relation is precisely one of the $\pi$-calculus~\cite{pi-calculus}, restricted to the $C_\pi$ processes.  One consequence of the non-forwarding property of our calculus is the possibility of the creation of closed domains for channels.
A property, resembling the creation of a secure channel, which scope is statically determined, can be formally stated using 
the definition of strong bisimilarity.

\begin{proposition}[Closed Domains for Channels]\label{prop:behavioural}
For any process $\PP$, channel $\NM$ and prefix $\pi$, the following equality holds

\[
\rest\NK
( (\rest\NL\send\NK\role\msg\NL.\receive\NM\role\msg\NY.\match\NY\role\msg\NL\pi.\inact) \parop \receive\NK\role\msg\NX.\PP) 
\sim 
\rest\NK
( (\rest\NL\send\NK\role\msg\NL.\receive\NM\role\msg\NY.\inact) \parop \receive\NK\role\msg\NX.\PP)
\]

\end{proposition}
\begin{proof}
The proof follows by coinduction on the definition of the strong bisimulation (see Appendix~\ref{app:proof1}).
\end{proof}

In both processes in Proposition~\ref{prop:behavioural} the left thread creates a new channel $\NL$ and sends it over a (private) channel $\NK$ to the right thread. 
The equality states that then the channel $\NL$ cannot be received afterward in the left thread. 
This is due to the fact that the right thread cannot forward received channels. 
We may notice that both processes in the proposition define the final scope for channel $\NL$, hence, determining a closed domain for the channel. 
The interpretation of this proposition can be twofold. 
On one hand, the right thread after receiving a fresh channel ($\NL$) cannot send the received channel, 
since it respects the non-forwarding 
property.
On the other hand, 
the left thread sends the channel $\NL$ (to the right thread) only once 
and the channel afterward behaves ``statically", since then it cannot be exchanged even between any two sub-processes of process $\PP\subst{\NL}{\NX}$. 
Further explanations are given in the next section.

\paragraph{The $\pi$-calculus processes that do not forward 
names.} The $C_\pi$ processes satisfy our non-forwarding 
property (Definition~\ref{def:non-leaking-pi-processes}): if the received channel is new to the process it will not be sent later on.  
Generally, the $\pi$-calculus processes do not meet the non-forwarding property.
However, we may notice that there are some $\pi$-calculus processes which are not part of the $C_\pi$ syntax but still respect this property. 
For example, consider the $\pi$-calculus process 
\[
\receive\NK\role\msg\NX.\rest\NL(\send\NL\role\msg\NX.\inact \parop \receive\NL\role\msg\NY.\inact)
\]
where any received channel along $\NK$ is then sent on $\NL$, but since $\NL$ is restricted, the process will not output the received channel. But the condition that the channel along which the forwarding is performed (here $\NL$) is restricted is not enough. For example, the $\pi$-calculus process 
$\receive\NK\role\msg\NX.\rest\NL(\send\NK\role\msg\NL.\send\NL\role\msg\NX.\inact \parop \receive\NL\role\msg\NY.\inact)$ does not satisfy the non-forwarding 
property.

This hints that differentiating $\pi$-calculus processes that do not forward 
received channels, in any of their possible evolutions, may not be a simple task.
As one solution to the problem we propose the next result which states that a $\pi$-calculus process $\PP$, and any of its possible evolutions, do not forward 
received channels if one can find a $C_\pi$ process $\PQ$, such that $\PP$ and $\PQ$ are bisimilar. 
In the theorem, we refer to the non-forwarding 
property of Definition~\ref{def:non-leaking-pi-processes}, extended to consider all $\pi$-calculus processes.
Naturally, the result of the next theorem refers to sum-free $\pi$-calculus processes, since in this paper we are not considering the sum operator in the $C_\pi$.

\begin{proposition}[The $\pi$-calculus Processes That do not Forward 
Names]\label{prop:non-forwarding-of-pi-processes}
Let $\PP$ be a 
$\pi$-calculus process. If there is a $C_\pi$ process $\PQ$, such that $\PP\sim\PQ$, then $\PP$ satisfies the non-forwarding 
property.

\end{proposition}

\begin{proof}
The proof is derived to Appendix~\ref{app:proof1}.
\end{proof}

Although the result of Proposition~\ref{prop:non-forwarding-of-pi-processes} is only of the existential nature, we believe it is a step towards more practical results. One such result might be proving that for a given $\pi$-calculus process $\PP$ one can derive a $C_\pi$ process $\PQ$ such that if $\PP\sim\PQ$ then $\PP$ respect the non-forwarding property. There we can also use a relaxed definition of the non-forwarding property, in which processes do not forward names received along some predefined set of channels.  We leave such investigations for future work.




\section{Examples}\label{sec:examples}

In this section, we further investigate some interesting scenarios representable in $C_\pi$. 
Since a process in our calculus can learn new names but cannot gain the capability to send such names, 
we may distinguish two levels of channel ownership of a process that are invariant to the process evolution: 
\begin{itemize}
\item [-] \emph{administrator}: the process that creates the channel, it
has all capabilities over the channel;
\item [-] \emph{user}: the process that learns the channel name through communication (scope extrusion) 
and can communicate along the channel but cannot send it.
\end{itemize}

Hence, all administrators are also users but the conversely is not true. 
Also, notice that any process that receives a channel can become a user for that channel (but not administrator), and, hence, all processes may be considered as potential users for any channel. 
If we consider modelling our opening example in $C_\pi$ calculus 
\[
\rest{\mathit{session}}\send{channel}\role\msg{\mathit{session}}.\mathit{Alice} \quad\parop\quad \receive{\mathit{channel}}\role\mgs{\NX}.\mathit{Bob} \quad\parop\quad \mathit{Carol}
\]
$\mathit{Alice}$ is the administrator for $\mathit{session}$ and $\mathit{Bob}$ becomes a user after the reception. 
In $C_\pi$ it is not possible for $\mathit{Bob}$ afterward to send $\mathit{session}$ to a third party (e.g., to $\mathit{Carol}$). If $\mathit{Bob}$ wants $\mathit{Carol}$ to get the access to channel $\mathit{session}$ he can only tell $\mathit{Alice}$ and let her decide whether she wants to send $\mathit{session}$ to $\mathit{Carol}$ or not. 
Therefore, we can have 
\begin{itemize}
\item $\mathit{Bob}=\send{\mathit{channel}}\role\msg{\mathit{carol}}.\inact$, where $\mathit{Bob}$ sends to $\mathit{Alice}$ channel $\mathit{carol}$, and then terminates;
\item $\mathit{Alice}=\receive{\mathit{channel}}\role\msg\NY.\send{\NY}\role\msg{\mathit{session}}.\mathit{Alice'}$, where $\mathit{Alice}$ receives the channel from $\mathit{Bob}$ and decides to send $\mathit{session}$ along the received channel, and
\item $\mathit{Carol}=\receive{\mathit{carol}}\role\msg\NX.\mathit{Carol'}$, where $\mathit{Carol}$ can finally  receive $\mathit{session}$ along channel $\mathit{carol}$.
\end{itemize} 
We remark this simple example relies on the purely concurrent setting, here $\mathit{session}$ can be held by three or more parties at the same time. The correlation of the $C_\pi$ and the linearity of session types would need further investigation.

\subsection{Authentication}\label{sec:authentication}
In the $C_\pi$-calculus specifying $\rest\NL\PP$ means $\PP$ is the administrator for channel $\NL$. Process $\PP$ can extrude the scope of $\NL$ by sending, but none of the receiving processes will ever become administrators for the received channel $\NL$, since they will never gain the capability for sending $\NL$. 
Since sending capability cannot be transferred to other processes, we may conclude the administrator property over some channel can be used as an authentication over
processes. 
For example, a process $\PQ$ that previously has received $\NL$ may check at any point if the other party with whom he is communicating at the moment over $\NL$ is actually an administrator or only a user for channel $\NL$. This can be done by specifying 
\[
\PQ=\rest\NN\send\NL\role\msg\NN.\receive\NN\role\msg\NX.\match\NX\role\msg\NL.\PQ',
\]
 where first the private session with other party listening on $\NL$ is established by sending a fresh channel $\NN$, and then along $\NN$ a channel is expected to be received. If the channel received is $\NL$, then the other party has proved to be an administrator for $\NL$. Notice that $\PQ$ itself does not have to be an administrator for $\NL$ in this example.

As another example, consider that the above process $\PP$ has two threads running in parallel
\[
\send\NK\role\msg\NL.\receive\NK\role\msg\NY.\match\NY\role\msg\NL\pi.\PP_1
\quad\parop\quad
\receive\NK\role\msg\NX.\match\NX\role\msg\NL\send\NK\role\msg\NL.\PP_2
\]
where, before activating $\pi.\PP_1$ and $\PP_2$ and their possible interactions, 
both threads test whether the other one is an administrator for channel $\NL$. 
Namely, after the first synchronization, the left thread matches the received channel with 
$\NL$, i.e., we obtain configuration 
$\receive\NK\role\msg\NY.\match\NY\role\msg\NL\pi.\PP_1 
\parop\match\NL\role\msg\NL\send\NK\role\msg\NL.\PP_2\subst\NL\NX$.
If the channel received is $\NL$, then the right thread concludes that the left thread is 
an administrator for $\NL$ and sends the same channel back. 
Then, the left thread also matches the received channel with $\NL$ and continues only if the two names match, 
leading to
$\match\NL\role\msg\NL\pi.\PP_1\subst\NL\NY 
\parop\PP_2\subst\NL\NX$.
After  that, both threads have proved to be administrators for $\NL$, meaning that both have proved that they originate from process $\PP$.

\subsection{Modelling groups and name hiding}\label{sec:groups}

Controlling name sharing in the $\pi$-calculus has been investigated in past and several process models are proposed to this end. In~\cite{cardelli05}, on the $\pi$-calculus syntax, an additional construct is introduced, called group creation, and a typing discipline is developed. The intention of the group construct is to restrict communications: channels specified to be in a group cannot be communicated outside the scope of the corresponding group construct. Hence, the group creation closes the domain for channels specified in the construct. In~\cite{giunti12}, on the $\pi$-calculus syntax, an additional construct hide is introduced. Construct hide has similar properties to channel restriction, but it is more static since it forbids the channel extrusion for the channel specified to be hidden. Again, roughly speaking, hide construct closes the domain for a channel specified in the construct. In what follows, we try to represent similar behaviors without any additional language constructs, but directly in the $C_\pi$ calculus (and hence, directly in the $\pi$-calculus). Formalization of the relationship between the mentioned models and the $C_\pi$ is left for future work.

As one way to represent closed domain for a channel $\NL$ in the $C_\pi$ we may consider process
%
\[
\rest\NK 
( \rest\NL\send\NK\role\msg\NL.\inact \parop \receive\NK\role\msg\NX.\PP)
\]
resembling the one from Proposition~\ref{prop:behavioural}, where the left thread creates the channel $\NL$, sends it to the right thread and then terminates. 
The right thread receives the channel, after which the two actions synchronize and the starting process 
silently evolves to $\rest\NK \rest\NL(\inact \parop \PP\subst\NL\NX)$. 
As we have shown in Proposition~\ref{prop:behavioural},  process (and any its subprocess) $\PP\subst\NL\NX$ cannot perform output with object $\NL$ 
since channel $\NL$ was not originally created by process $\PP$. 
For this property we can state that name $\NL$ will never leak out of the scope of $\PP$, hence that all 
communications along channel $\NL$ are private to process $\PP$. 
Also, we may observe that channel $\NL$ in process $\rest\NK \rest\NL(\inact \parop \PP\subst\NL\NX)$ 
has a 
static behavior, similar to $CCS$ channels~\cite{DBLP:books/sp/Milner80}.

This constellation indeed resembles the group creation of the 
$\pi$-calculus with groups and name hiding. 
The similarity is that in the $\pi$-calculus with groups, a channel declared as 
a member of a group cannot be acquired as 
a result of communication by the process outside the scope of the group. 
The major difference is that in our example the channel behaves as a $CCS$-like channel, 
i.e., the channel cannot be acquired as a result of communication by any process. 
This brings us to our next example, that combines channel declaration and authentication, 
presented in Section~\ref{sec:authentication}. 
Consider that for a given channel, we want to statically determine a boundary for the possible channel 
extrusion, the part of the process we shall call a group. 
In that case, we may conclude that only members of the group should be 
able to receive the given channel. 
As a concrete example consider process
\[
\rest{\NGG_\NL} (\rest\NL\PP \parop \PQ)
\]
where by $\NGG_\NL$ we denote that the scope of channel $\NGG_\NL$ determines the group for channel $\NL$. 
Now, to make sure that channel $\NL$, whose administrator is process $\PP$, is sent only 
to processes scoped over with $\NGG_\NL$, before each sending of channel $\NL$, 
we must make sure that the receiver is an administrator for channel $\NGG_\NL$. 
Hence, instead of construct $\send\NK\role\msg\NL$ in $\PP$, we would use 
\[
\rest\NN \send\NK\role\msg\NN. \receive\NN\role\msg\NX. \match\NX\role\msg{\NGG_\NL} \send\NN\role\msg\NL
\]
where first a private session with the process willing to receive $\NL$ 
is established through channel $\NN$, and then the
channel received on $\NN$ is matched with $\NGG_\NL$. 
Only if the name received on $\NN$ is $\NGG_\NL$, i.e., after the other process has proved 
that he is a member of the group, channel $\NL$ is sent. Notice that if the other process can send $\NGG_\NL$ this  means the process originates either from $\PP$ or from $\PQ$.

\subsection{Open-ended groups}\label{sec:open-ended-groups}

Groups described in the previous example provide an interesting framework to investigate 
sharing protected resources in distributed environments but has the limitation that a 
group, once created, always has a fixed scope, which sometimes might be considered too restrictive. 
In $C_\pi$ 
open-ended groups are directly modeled, since sending a resource in $C_\pi$ does not transmit 
the capability for its further dissemination. 
For example, in $\rest{\mathit{group}}\send\NK\role\msg\mathit{group}.\PP$ the administrator of $\mathit{group}$, that is the process that creates the group, can send the name of the group to other processes, 
while the receiving process only becomes a user of the group and does not gain the capability to 
invite new members to the group.

\section{Encoding Uncontrolled Name Passing}\label{sec:encoding}

In this section, we show how to model forwarding in $C_\pi$, 
as in the standard $\pi$-calculus. 
We start by presenting the basic idea, which we later formalize by means of an encoding.

Throughout this 
section, we use the polyadic version of our calculus. 
This enables us to formalize our ideas in a more crisp way, but, as in the $\pi$-calculus, 
each polyadic communication 
%
%
can be represented by a sequence of monadic ones.
%
Furthermore, we will use poliadicity in a controlled way, so that we do not introduce a non-well sorted communications~\cite{DBLP:conf/concur/Milner92}.

In the $\pi$-calculus, there is no syntactic restriction on names that can appear as 
objects of output prefixes, hence a process like 
$
\receive\NK\role\msg\NX.\send\NGG\role\msg\NX.\PP  \parop \send\NK\role\msg\NL.\PQ
$
%
may be specified.
One way to represent this process in $C_\pi$ is to:
\begin{itemize}
\item create a special process dedicated for (repeatedly) sending channel $\NL$, called \emph{handler} of the channel, 
\item while sending channel $\NL$ also send a channel dedicated for communicating with the handler, and 
\item bypass sending of the received channel  to the handler process.
\end{itemize}  
Hence, we may try to represent the $\pi$-calculus process introduced above as
\[
 \receive\NK\role\msg{(\NX_1, \NX_2)}.\send{\NX_2}\role\msg\NGG.\PP  
 \parop \send\NK\role\msg{(\NL,\NM_\NL)}.\PQ
 \parop \receive{\NM_\NL}\role\msg\NY.\send\NY\role\msg{(\NL,\NM_\NL)}.\inact  
\]
where the process in the middle now sends $\NL$ together with channel $\NM_\NL$ dedicated for 
communicating with the handler process 
(the rightmost one), and the leftmost process receives both channels and instead of sending 
$\NL$ along $\NGG$ it sends $\NGG$ to the handler along $\NM_\NL$. 
The handler process receives $\NGG$ 
and sends $\NL$ (again, together with $\NM_\NL$) along the received channel, in such way mimicking forwarding.
Such representation does not work in the case when the leftmost process is not an administrator 
for channel $\NGG$ (e.g., assume process $\receive\NK\role\msg\NX.\send\NGG\role\msg\NX.\PP$ is derived from $\receive\NK\role\msg\NY.\receive\NK\role\msg\NX.\send\NY\role\msg\NX.\PP$), since then the process cannot send $\NGG$ to the handler process. 
To this end, we must refine our representation of forwarding to support situations 
when processes are potentially not administrators for any given channel. 
Thus, we introduce another type of handler processes which are in charge of forwarding channels 
that are subjects of output actions. 
For example, the starting $\pi$-calculus process would be represented as
\[
 \receive\NK\role\msg{(\NX_1, \NX_2)}.\rest\NE \send{\NN_{\NGG}}\role\msg\NE.\send{\NX_2}\role\msg\NE.\PP  
 \parop \send\NK\role\msg{(\NL, \NM_{\NL})}.\PQ  
 \parop \receive{\NM_{\NL}}\role\msg\NY.\receive\NY\role\msg\NZ.\send\NZ\role\msg{(\NL,\NM_{\NL} )}.\inact
 \parop \receive{\NN_{\NGG}}\role\msg\NY.\send\NY\role\msg{\NGG}.\inact
\]
where we added the rightmost thread, which is the handler process of channel $\NGG$, used to bypass 
sending of channel $\NGG$ in the leftmost thread. 
Now the communication in this process goes as follows: first, the two leftmost threads synchronize on channel $\NK$ (as in previous examples), leading to 
\[
 \rest\NE \send{\NN_{\NGG}}\role\msg\NE.\send{\NM_{\NL}}\role\msg\NE.\PP
 \parop \PQ  
 \parop \receive{\NM_{\NL}}\role\msg\NY.\receive\NY\role\msg\NZ.\send\NZ\role\msg{(\NL,\NM_{\NL} )}.\inact
 \parop \receive{\NN_{\NGG}}\role\msg\NY.\send\NY\role\msg{\NGG}.\inact
\]
where the name $\NX_2$ is instantiated with handler name $\NM_{\NL}$ (eliding from the substitution in process $\PP$). Then, instead of sending $\NL$ on $\NGG$, the new channel $\NE$ is created and 
sent to the handlers of $\NL$ and $\NGG$
\[
 \rest\NE( \PP  
 \parop \PQ  
 \parop \receive\NE\role\msg\NZ.\send\NZ\role\msg{(\NL,\NM_{\NL} )}.\inact
 \parop \send\NE\role\msg{\NGG}.\inact)
\]
enabling handler of channel $\NGG$ to send 
$\NGG$ to the handler of channel $\NL$, leading to 
$
 \rest\NE( \PP  
 \parop \PQ  
 \parop \send\NGG\role\msg{(\NL,\NM_{\NL} )}.\inact
 \parop \inact)
$
where sending $\NL$ (together with the handling name $\NM_{\NL}$) 
along channel $\NGG$ is finally activated.

Notice that we need a few generalizations of this approach:
\begin{itemize}
\item each name can be sent infinitely many times, hence, each handling process must be 
repeatedly available for communication;
\item each channel can be used either as a subject or as an object of output action, 
and, hence, for each channel we need both types of handler processes 
(one as for $\NL$ and the other as for $\NGG$ in the last example). For the rest of this section we will call the handler of a channel a process that comprehends both types of handlers mentioned above for that channel;
\item in the source language, processes synchronize their dual (i.e., input/output) actions directly, while in the target language output process first synchronize with the handler process, and only after that handler process synchronize with the input process. In the example above, process $\PP$ may proceed even though the channel $\NL$ has not been received by any process. Hence, we need a mechanism to allow processes in the target language to synchronize their actions directly.
\end{itemize}

\begin{table}
\[
\begin{array}[t]{@{}r@{}c@{}l}
\enc{ \rest \NK \PP} &=& 
\rest{\NK, \NN_\NK, \NM_\NK} 
(\enc{ \PP }
\parop
\rep\receive{\NN_\NK}\role\msg\NX.\send\NX\role\msg\NK .\inact
\parop 
\rep\receive{\NM_\NK}\role\msg{(\NX_1, \NX_2)}.\receive{\NX_1}\role\msg\NY. \rest\NT\send{\NY}\role\msg{(\NK, \NN_\NK, \NM_\NK, \NT)}.\send{\NX_2}\role\msg\NT.\inact )
\vspace{2ex}\\
\enc{ \send\NA\role\msg\NB.\PP }&= & 
\rest{\NE_1,\NE_2} \send{\NN_\NA}\role\msg{\NE_1}.\send{\NM_\NB}\role\msg{(\NE_1, \NE_2)}.\receive{\NE_2}\role\msg\NY.\send\NY\role\msg{\NE_1}.  \enc{ \PP } 
\vspace{2ex}\\
\enc{ \receive\NA\role\msg\NX.\PP } &=& 
\receive\NA\role\msg{(\NX, \NN_\NX, \NM_\NX, \NX')}.\receive{\NX'}\role\msg\NY. \enc{ \PP }
\vspace{2ex}\\
\enc{ \match{\NC_1}\role\msg{\ND_1}\ldots\match{\NC_n}\role\msg{\ND_n}\receive\NA\role\msg\NX.\PP } &=& 
\match{\NC_1}\role\msg{\ND_1}\ldots\match{\NC_n}\role\msg{\ND_n}\receive\NA\role\msg{(\NX, \NN_\NX, \NM_\NX, \NX')}.\receive{\NX'}\role\msg\NY. \enc{ \PP }
\vspace{2ex}\\
\enc{ \match{\NC_1}\role\msg{\ND_1}\ldots\match{\NC_n}\role\msg{\ND_n}\send\NA\role\msg\NB.\PP }
&=& 
\rest{\NE_1,\NE_2} \match{\NC_1}\role\msg{\ND_1}\ldots\match{\NC_n}\role\msg{\ND_n} \send{\NN_\NA}\role\msg{\NE_1}.\send{\NM_\NB}\role\msg{(\NE_1, \NE_2)}.\receive{\NE_2}\role\msg\NY.\send\NY\role\msg{\NE_1}. \enc{ \PP }
\end{array}\\
\]
\[
\begin{array}[t]{@{}c@{}}
\enc{ \PP_1 \parop \PP_2 } = 
\enc{ \PP_1} \parop \enc{ \PP_2 }
\qquad
\enc{ \rep\PP } = \rep \enc{ \PP }
\qquad
\enc{ \inact} = \inact
\end{array}
\]
\caption{\label{tab:Encoding} Encoding of $\pi$-calculus processes into $C_\pi$ processes.}
\end{table}

We formalize these ideas by introducing an encoding as a pair $(\enc{\cdot}, \varphi_{\enc{\;}})$, where $\enc{\cdot}$ is a translation function and $\varphi_{\enc{\;}}$ is a renaming policy~\cite{DBLP:journals/iandc/Gorla10}. The translation maps each $\pi$-calculus (source) term $\PP$ into the $C_\pi$ (target) term $\enc{\PP}$, and while doing so it uses the renaming policy, that maps each name of the source term $\NA$ into a tuple of names $(\NA, \NN_\NA, \NM_\NA)$, where $\NN_\NA$ and $\NM_\NA$ are not names of any source term, and for each name $\NA$ different names $\NN_\NA$ and $\NM_\NA$ are used. Also, the translation uses some  additional names, which we assume to be from a reserved set of names, disjoint from all names of the source language and all names introduced by the renaming policy. All these definitions follow the idea of~\cite{DBLP:journals/iandc/Gorla10}.

The translation function is defined in Table~\ref{tab:Encoding}.
%
%
The first rule in the table translates a process scoped with the channel restriction. The source process is encoded as 
scoped with the original channel name $\NK$ and the two names associated to $\NK$ by the renaming policy, i.e. $\NN_\NK$ and $\NM_\NK$, and it introduces the handler process for the channel in parallel with $\enc{\PP}$.
We use $\rest{\NK, \NN, \NM}$ to abbreviate $\rest\NK \rest \NN\rest \NM$. 
The handler process has two threads in parallel. The left thread is repeatedly available to be invoked on $\NN_\NK$ (see the rule of output) and it sends channel $\NK$ along the received channel. 
The use of this process will be only to send $\NK$ to the right thread of any other handler process (as $\NN_\NGG$ in the example above). 
The right thread of the handler is repeatedly available to be invoked on $\NM_\NK$ and it receives a pair of names (see the rule for output). 
Along the left received name ($\NX_1$) it receives a channel (from the left thread of some handler process) and outputs the channel $\NK$ together with ``addresses"  of the handler, $\NN_\NK$ and $\NM_\NK$, and, in addition, a new channel $\NT$. 
By sending $\NN_\NK$ and $\NM_\NK$, we make possible for the process that receives (see the rule for input) to be able afterward to directly invoke the handler for $\NK$.
By sending a new channel $\NT$ to the receiving process and also (in the continuation of the right thread of the handler) to the sending process we establish a private connection between the two processes, which then can directly synchronize and activate their continuations at the same time. 


The encoding of the output process creates two fresh channels $\NE_1$ and $\NE_2$, and sends one end of $\NE_1$ to the left thread handler process 
of name $\NA$ (the subject of the output) and the other end of $\NE_1$, together with and $\NE_2$, to the right thread of the handler process for
name $\NB$ (the object of the output). In the continuation, before activating the image of $\PP$, along $\NE_2$ a channel is received and used for the output (to synchronize directly with the input process). 
We remark that in this rule names $\NE_1, \NE_2$ and $\NY$ are taken to be from the reserved set of names, and hence cannot appear as free in  $\enc\PP$. 
The same assumption is made for names $\NX'$ and $\NY$ in the rule for input, hence there also $\NX'$ and $\NY$ are not free in $\enc\PP$. 
In the rule for input four channels are received, the channel together with addresses of his handler and a fresh channel (see the right thread of a handler process). The received fresh channel is only used, as we noted, to synchronize with the sending process (same as the channel received in this synchronization).
The rest of the rules shows that the encoding is homomorphism elsewhere. 

Notice the encoding does not interfere with our notion of ownership described in Section~\ref{sec:examples}. The role of channel administrator is still present, i.e., handlers are included in that domain. 
Hence, controlling such domain can still be done, in contrast to the regular $\pi$-calculus processes where one cannot statically identify a domain where the sending the channel capability is confined to.

We may also notice that in the rule for output (Table~\ref{tab:Encoding}) the two channels (the addresses) of the two handler processes are used, one for the object and the other for the subject of the prefix. 
This reflects the fact that in order to be capable to mimic all the actions of the source term, we need to introduce handler processes for all free names of the input and output prefixes of the source term. 
Since the handler processes are introduced directly in the rule for restricted channels (Table~\ref{tab:Encoding}), we give our main result for the correctness of the encoding only for the $\pi$-calculus processes that contain only bound names. 
A $\pi$-calculus process that has no free names 
is called \emph{closed}.

As closed $\pi$-calculus processes can only exhibit $\tau$ transitions, and those match the reduction semantic~\cite{pi-calculus}, for the simplicity we chose to deal with the reduction semantics of the $\pi$-calculus. 
Therefore, our operational correspondence result relates the set of closed $\pi$-calculus processes, with the reduction semantics (using the reduction relation $\red$, as defined in~\cite{pi-calculus}), and $C_\pi$-processes with the labeled transition system.
Notice that the reduction relation of the $\pi$-calculus relies on the structural congruence relation~\cite{pi-calculus}, which, by rule $\match\NA\role\msg\NA\pi.\PP\equiv\pi.\PP$, may introduce free names. 
These newly introduced names are not of our interest, as they do not require handlers. 
To this end, for a $\pi$-calculus process $\PP$ we define $\fnn{\PP}$, a subset of $\fn{\PP}$ that is invariant with respect to the structural congruence relation. Hence, we define $\fnn{\match\NA\role\msg\NA\pi.\PP}=\fnn{\pi.\PP}$, and otherwise $\fnn{\PP}$ coincides with $\fn{\PP}$. Notice that this means $\fnn{\match\NA\role\msg\NB\pi.\PP}=\{\NA,\NB\}\cup\fnn{\pi.\PP}$, if $\NA\not=\NB$.
 
We use $\ltss{\tau}*$ to denote the transitive closure of $\lts{\tau}$.
We are now ready to present our result, showing that if the source term $\PP$ reduces to $\PQ$ then the encoding of $\PP$ reduces in a number of $\tau$ steps to the process bisimilar to the encoding of the process $\PQ$. Again, we assume that the $\pi$-calculus processes are sum-free.
First, we present an auxiliary result.

\begin{lemma}\label{lem:oper-corresp-with-Hs}
If $\PP \red \PQ$ then 
$
\enc{\PP}\parop \handler \ltss{\tau}* \sim \enc{\PQ}\parop \handler,
$
where
\begin{itemize}
\item if $\fnn{\PP}=\{\NK_1, \ldots, \NK_n\}$ then 
\[
\handler=\prod\limits_{i\in\{1, \ldots, n\}} \handler_{\NK_i},
\]
\item if $\fnn{\PP}=\emptyset$ then 
$\handler=\inact$.
\end{itemize}
\end{lemma}
\begin{proof}
The proof is derived to Appendix~\ref{app:proof2}.
\end{proof}
The above result can already be seen as a form of non-standard completeness as it uses the top-level handlers. This hints that the completeness result can also be stated for $\pi$-calculus processes that are not closed, with the encoding that is not compositional, but weakly compositional, such as the encoding from the join-calculus into the $\pi$-calculus~\cite{DBLP:conf/popl/FournetG96}. 
We leave such investigations for future work.

Our main result, given in the next theorem, is a direct consequence of Lemma~\ref{lem:oper-corresp-with-Hs}.

\begin{theorem}[Operational Correspondence: Completeness]\label{theorem:operational-corresp}
Let $\PP$ be a closed (sum-free) $\pi$-calculus process. 
If $\PP\red\PQ$ then $\enc{\PP}\ltss{\tau}* \sim\enc{\PQ}$. 
\end{theorem}

We also believe that our encoding satisfies the soundness property~\cite{DBLP:journals/iandc/Gorla10}. 
This claim relies on the fact that the encoding manipulates the source names in a controlled way. Each name of a source term is translated into a triple of names, and a renaming policy ensures that for different source names different triples of names are used. Other names introduced by the encoding are bound. Using this fact and the fact that each source prefix is translated into a sequence of prefixes of always the same length, we may notice that different post-processing steps might get interleaved, but post-processing steps of different reduction steps in a source term cannot interfere with each other. We also leave formalization of this claim for future work.


\section{Conclusions and related work}\label{sec:related work}

The notion of secrecy has been studied intensively in process calculi in the past 
and the variety of techniques have been proposed. 
The most related to our work are process models building on the $\pi$-calculus, 
such as~\cite{cardelli05, giunti12, crafa07, DBLP:journals/lmcs/KouzapasP17, hennessy05, vivas02}.

Cardelli et al.~\cite{cardelli05} introduce a language construct for 
group creation and a typing discipline, where a group is a type for a channel. 
The group creation construct blocks communications 
of channels that are declared as members of the group outside the initial scope of the group, 
hence preventing the leakage of protected channels. 
Kouzapas and Philippou~\cite{DBLP:journals/lmcs/KouzapasP17} extend 
the model of $\pi$-calculus with groups by constructs that allow reasoning about the private 
data in information systems.
The work of Giunti et. al.~\cite{ giunti12} introduces an operator called hide which 
binds a name and has a similar behavior as a name restriction, but in contrast to 
name restriction it blocks a name extrusion, for which the scope of the hide operator forms 
a kind of a group that the ``hidden" name cannot exit. 
The paper by Vivas and Yoshida~\cite{vivas02} introduces an operator called filter that 
is statically associated to a process and blocks all actions of the process along names that are 
not contained in the (polarized) filter.
We also mention~\cite{hennessy05, crafa07} where the types associate the 
security levels to channels, where, in the latter work downgrading the security level of a channel
is admissible and it is achieved by introducing special, so-called, declassified 
input and output prefix constructs. 
All the above approaches share the property that, when building on the $\pi$-calculus model, 
additional language construct and/or a typing discipline is introduced 
in order to represent some specific aspect of secrecy 
in a dedicated way. 
We believe that $C_\pi$-calculus appears to be more suitable as an underlying model when studying 
secrecy, and as such that many aspects of secrecy can be represented in a more canonical way. 
As a first step, we plan to make a precise representation of group creation~\cite{cardelli05} in the 
$C_\pi$-calculus, following the intuition provided in Section~\ref{sec:groups}.

Several fragments of the $\pi$-calculus have been used 
in different ways and for different purposes. 
The asynchronous $\pi$-calculus~\cite{DBLP:conf/ecoop/HondaT91}, proposed by Honda and Takoro,
constrains the syntax by allowing only an inactive process to be the continuation of the output 
prefix, in this way modelling asynchronous communications. 
The Localised $\pi$-calculus~\cite{merro04}, proposed by Merro and Sangiorgi, 
disallows the input capability for the received names and does not consider the matching operator. 
There, the syntactic restriction is that input placeholder cannot appear as a subject of an 
input, but, in contrast to our work, the forwarding of names is allowed.
The Private $\pi$-calculus~\cite{DBLP:journals/tcs/Sangiorgi96a}, proposed by Sangiorgi, 
makes the restriction that objects of output prefixes are always considered as bound, 
making the symmetry with the input prefixes. 
Although in Private $\pi$-calculus the forwarding of names is not possible, 
it differs significantly from our work in the restriction that one name can be sent only once. 
All these calculi share our goal to investigate specific notions in a dedicated way, without requiring the introduction of specialized primitives, instead by considering a suitable fragment of the $\pi$-calculus.

In this paper, we have presented Confidential $\pi$-calculus, 
a fragment of the $\pi$-calculus~\cite{pi-calculus} in which the forwarding of 
received names is disabled at the syntax level. 
To the best of our knowledge, this is the first process model based on the $\pi$-calculus 
that 
represents the controlled name passing by constraining and not extending the original syntax. 
Some specific properties of our labeled transition system are given and the non-forwarding property is defined. All $C_\pi$ processes satisfy this property and the method to differentiate the $\pi$-calculus process that never forward received names is proposed, relying on the strong bisimilarity relation of $C_\pi$ processes and $\pi$ processes. The strong bisimilarity relation is also used to show that the creation of closed domains for channels is directly representable in the $C_\pi$.
Examples presented in the paper already give some intuition 
on scenarios directly representable in $C_\pi$, such as authentication and group modelling, 
and a complete formalization of these ideas is left for future work.
The encoding presented here shows that our model is as expressive as the $\pi$-calculus, 
and the formal verification of the correctness of the encoding is given in the form of the completeness of  operational correspondence. The soundness of the encoding is left for future work.

{\bf Acknowledgments.} 
The author would like to thank Hugo Torres Vieira and Jovanka Pantovi\' c for supervising the work presented here and a great help, to Daniel Hirschkoff for valuable email discussions, and to anonymous reviewers for useful remarks and suggestions. This work has been partially supported by the Ministry of Education and Science of the Republic of Serbia, project ON174026.


\bibliographystyle{eptcs}
\bibliography{bib}

\begin{thebibliography}{10}
\providecommand{\bibitemdeclare}[2]{}
\providecommand{\surnamestart}{}
\providecommand{\surnameend}{}
\providecommand{\urlprefix}{Available at }
\providecommand{\url}[1]{\texttt{#1}}
\providecommand{\href}[2]{\texttt{#2}}
\providecommand{\urlalt}[2]{\href{#1}{#2}}
\providecommand{\doi}[1]{doi:\urlalt{http://dx.doi.org/#1}{#1}}
\providecommand{\bibinfo}[2]{#2}

\bibitemdeclare{article}{cardelli05}
\bibitem{cardelli05}
\bibinfo{author}{Luca \surnamestart Cardelli\surnameend},
  \bibinfo{author}{Giorgio \surnamestart Ghelli\surnameend} \&
  \bibinfo{author}{Andrew~D. \surnamestart Gordon\surnameend}
  (\bibinfo{year}{2005}): \emph{\bibinfo{title}{Secrecy and group creation}}.
\newblock {\sl \bibinfo{journal}{Inf. Comput.}}
  \bibinfo{volume}{196}(\bibinfo{number}{2}), pp. \bibinfo{pages}{127--155},
  \doi{10.1016/j.ic.2004.08.003}.

\bibitemdeclare{article}{crafa07}
\bibitem{crafa07}
\bibinfo{author}{Silvia \surnamestart Crafa\surnameend} \&
  \bibinfo{author}{Sabina \surnamestart Rossi\surnameend}
  (\bibinfo{year}{2007}): \emph{\bibinfo{title}{Controlling information release
  in the pi-calculus}}.
\newblock {\sl \bibinfo{journal}{Inf. Comput.}}
  \bibinfo{volume}{205}(\bibinfo{number}{8}), pp. \bibinfo{pages}{1235--1273},
  \doi{10.1016/j.ic.2007.01.001}.

\bibitemdeclare{inproceedings}{DBLP:conf/popl/FournetG96}
\bibitem{DBLP:conf/popl/FournetG96}
\bibinfo{author}{C{\'{e}}dric \surnamestart Fournet\surnameend} \&
  \bibinfo{author}{Georges \surnamestart Gonthier\surnameend}
  (\bibinfo{year}{1996}): \emph{\bibinfo{title}{The Reflexive {CHAM} and the
  Join-Calculus}}.
\newblock In \bibinfo{editor}{Hans{-}Juergen \surnamestart Boehm\surnameend} \&
  \bibinfo{editor}{Guy L.~Steele \surnamestart Jr.\surnameend}, editors: {\sl
  \bibinfo{booktitle}{Conference Record of POPL'96: The 23rd {ACM}
  {SIGPLAN-SIGACT} Symposium on Principles of Programming Languages, Papers
  Presented at the Symposium, St. Petersburg Beach, Florida, USA, January
  21-24, 1996}}, \bibinfo{publisher}{{ACM} Press}, pp.
  \bibinfo{pages}{372--385}, \doi{10.1145/237721.237805}.

\bibitemdeclare{article}{DBLP:journals/acta/GayH05}
\bibitem{DBLP:journals/acta/GayH05}
\bibinfo{author}{Simon~J. \surnamestart Gay\surnameend} \&
  \bibinfo{author}{Malcolm \surnamestart Hole\surnameend}
  (\bibinfo{year}{2005}): \emph{\bibinfo{title}{Subtyping for session types in
  the pi calculus}}.
\newblock {\sl \bibinfo{journal}{Acta Inf.}}
  \bibinfo{volume}{42}(\bibinfo{number}{2-3}), pp. \bibinfo{pages}{191--225},
  \doi{10.1007/s00236-005-0177-z}.

\bibitemdeclare{inproceedings}{giunti12}
\bibitem{giunti12}
\bibinfo{author}{Marco \surnamestart Giunti\surnameend},
  \bibinfo{author}{Catuscia \surnamestart Palamidessi\surnameend} \&
  \bibinfo{author}{Frank~D. \surnamestart Valencia\surnameend}
  (\bibinfo{year}{2012}): \emph{\bibinfo{title}{Hide and New in the
  Pi-Calculus}}.
\newblock In \bibinfo{editor}{Bas \surnamestart Luttik\surnameend} \&
  \bibinfo{editor}{Michel~A. \surnamestart Reniers\surnameend}, editors: {\sl
  \bibinfo{booktitle}{Proceedings Combined 19th International Workshop on
  Expressiveness in Concurrency and 9th Workshop on Structured Operational
  Semantics, {EXPRESS/SOS} 2012, Newcastle upon Tyne, UK, September 3, 2012.}},
  {\sl \bibinfo{series}{{EPTCS}}}~\bibinfo{volume}{89}, pp.
  \bibinfo{pages}{65--79}, \doi{10.4204/EPTCS.89.6}.

\bibitemdeclare{article}{DBLP:journals/iandc/Gorla10}
\bibitem{DBLP:journals/iandc/Gorla10}
\bibinfo{author}{Daniele \surnamestart Gorla\surnameend}
  (\bibinfo{year}{2010}): \emph{\bibinfo{title}{Towards a unified approach to
  encodability and separation results for process calculi}}.
\newblock {\sl \bibinfo{journal}{Inf. Comput.}}
  \bibinfo{volume}{208}(\bibinfo{number}{9}), pp. \bibinfo{pages}{1031--1053},
  \doi{10.1016/j.ic.2010.05.002}.

\bibitemdeclare{article}{hennessy05}
\bibitem{hennessy05}
\bibinfo{author}{Matthew \surnamestart Hennessy\surnameend}
  (\bibinfo{year}{2005}): \emph{\bibinfo{title}{The security pi-calculus and
  non-interference}}.
\newblock {\sl \bibinfo{journal}{J. Log. Algebr. Program.}}
  \bibinfo{volume}{63}(\bibinfo{number}{1}), pp. \bibinfo{pages}{3--34},
  \doi{10.1016/j.jlap.2004.01.003}.

\bibitemdeclare{inproceedings}{DBLP:conf/ecoop/HondaT91}
\bibitem{DBLP:conf/ecoop/HondaT91}
\bibinfo{author}{Kohei \surnamestart Honda\surnameend} \&
  \bibinfo{author}{Mario \surnamestart Tokoro\surnameend}
  (\bibinfo{year}{1991}): \emph{\bibinfo{title}{An Object Calculus for
  Asynchronous Communication}}.
\newblock In \bibinfo{editor}{Pierre \surnamestart America\surnameend}, editor:
  {\sl \bibinfo{booktitle}{ECOOP'91 European Conference on Object-Oriented
  Programming, Geneva, Switzerland, July 15-19, 1991, Proceedings}}, {\sl
  \bibinfo{series}{Lecture Notes in Computer Science}} \bibinfo{volume}{512},
  \bibinfo{publisher}{Springer}, pp. \bibinfo{pages}{133--147},
  \doi{10.1007/BFb0057019}.

\bibitemdeclare{inproceedings}{DBLP:conf/esop/HondaVK98}
\bibitem{DBLP:conf/esop/HondaVK98}
\bibinfo{author}{Kohei \surnamestart Honda\surnameend},
  \bibinfo{author}{Vasco~Thudichum \surnamestart Vasconcelos\surnameend} \&
  \bibinfo{author}{Makoto \surnamestart Kubo\surnameend}
  (\bibinfo{year}{1998}): \emph{\bibinfo{title}{Language Primitives and Type
  Discipline for Structured Communication-Based Programming}}.
\newblock In \bibinfo{editor}{Chris \surnamestart Hankin\surnameend}, editor:
  {\sl \bibinfo{booktitle}{Programming Languages and Systems - ESOP'98, 7th
  European Symposium on Programming, Held as Part of the European Joint
  Conferences on the Theory and Practice of Software, ETAPS'98, Lisbon,
  Portugal, March 28 - April 4, 1998, Proceedings}}, {\sl
  \bibinfo{series}{Lecture Notes in Computer Science}} \bibinfo{volume}{1381},
  \bibinfo{publisher}{Springer}, pp. \bibinfo{pages}{122--138},
  \doi{10.1007/BFb0053567}.

\bibitemdeclare{article}{DBLP:journals/lmcs/KouzapasP17}
\bibitem{DBLP:journals/lmcs/KouzapasP17}
\bibinfo{author}{Dimitrios \surnamestart Kouzapas\surnameend} \&
  \bibinfo{author}{Anna \surnamestart Philippou\surnameend}
  (\bibinfo{year}{2017}): \emph{\bibinfo{title}{Privacy by typing in the
  {\(\pi\)}-calculus}}.
\newblock {\sl \bibinfo{journal}{Logical Methods in Computer Science}}
  \bibinfo{volume}{13}(\bibinfo{number}{4}), pp. \bibinfo{pages}{1--42},
  \doi{10.23638/LMCS-13(4:27)2017}.

\bibitemdeclare{article}{merro04}
\bibitem{merro04}
\bibinfo{author}{Massimo \surnamestart Merro\surnameend} \&
  \bibinfo{author}{Davide \surnamestart Sangiorgi\surnameend}
  (\bibinfo{year}{2004}): \emph{\bibinfo{title}{On asynchrony in name-passing
  calculi}}.
\newblock {\sl \bibinfo{journal}{Mathematical Structures in Computer Science}}
  \bibinfo{volume}{14}(\bibinfo{number}{5}), pp. \bibinfo{pages}{715--767},
  \doi{10.1017/S0960129504004323}.

\bibitemdeclare{book}{DBLP:books/sp/Milner80}
\bibitem{DBLP:books/sp/Milner80}
\bibinfo{author}{Robin \surnamestart Milner\surnameend} (\bibinfo{year}{1980}):
  \emph{\bibinfo{title}{A Calculus of Communicating Systems}}.
\newblock {\sl \bibinfo{series}{Lecture Notes in Computer
  Science}}~\bibinfo{volume}{92}, \bibinfo{publisher}{Springer},
  \doi{10.1007/3-540-10235-3}.

\bibitemdeclare{inproceedings}{DBLP:conf/concur/Milner92}
\bibitem{DBLP:conf/concur/Milner92}
\bibinfo{author}{Robin \surnamestart Milner\surnameend} (\bibinfo{year}{1992}):
  \emph{\bibinfo{title}{The Polyadic Pi-calculus (Abstract)}}.
\newblock In \bibinfo{editor}{Rance \surnamestart Cleaveland\surnameend},
  editor: {\sl \bibinfo{booktitle}{{CONCUR} '92, Third International Conference
  on Concurrency Theory, Stony Brook, NY, USA, August 24-27, 1992,
  Proceedings}}, {\sl \bibinfo{series}{Lecture Notes in Computer Science}}
  \bibinfo{volume}{630}, \bibinfo{publisher}{Springer}, p.~\bibinfo{pages}{1},
  \doi{10.1007/BFb0084778}.

\bibitemdeclare{article}{DBLP:journals/tcs/Sangiorgi96a}
\bibitem{DBLP:journals/tcs/Sangiorgi96a}
\bibinfo{author}{Davide \surnamestart Sangiorgi\surnameend}
  (\bibinfo{year}{1996}): \emph{\bibinfo{title}{pi-Calculus, Internal Mobility,
  and Agent-Passing Calculi}}.
\newblock {\sl \bibinfo{journal}{Theor. Comput. Sci.}}
  \bibinfo{volume}{167}(\bibinfo{number}{1{\&}2}), pp.
  \bibinfo{pages}{235--274}, \doi{10.1016/0304-3975(96)00075-8}.

\bibitemdeclare{book}{pi-calculus}
\bibitem{pi-calculus}
\bibinfo{author}{Davide \surnamestart Sangiorgi\surnameend} \&
  \bibinfo{author}{David \surnamestart Walker\surnameend}
  (\bibinfo{year}{2001}): \emph{\bibinfo{title}{The Pi-Calculus - a theory of
  mobile processes}}.
\newblock \bibinfo{publisher}{Cambridge University Press}.

\bibitemdeclare{article}{solove2005taxonomy}
\bibitem{solove2005taxonomy}
\bibinfo{author}{Daniel~J \surnamestart Solove\surnameend}
  (\bibinfo{year}{2005}): \emph{\bibinfo{title}{A taxonomy of privacy}}.
\newblock {\sl \bibinfo{journal}{U. Pa. L. Rev.}} \bibinfo{volume}{154}, p.
  \bibinfo{pages}{477}, \doi{10.2307/40041279}.

\bibitemdeclare{inproceedings}{DBLP:conf/fm/TschantzW09}
\bibitem{DBLP:conf/fm/TschantzW09}
\bibinfo{author}{Michael~Carl \surnamestart Tschantz\surnameend} \&
  \bibinfo{author}{Jeannette~M. \surnamestart Wing\surnameend}
  (\bibinfo{year}{2009}): \emph{\bibinfo{title}{Formal Methods for Privacy}}.
\newblock In \bibinfo{editor}{Ana \surnamestart Cavalcanti\surnameend} \&
  \bibinfo{editor}{Dennis \surnamestart Dams\surnameend}, editors: {\sl
  \bibinfo{booktitle}{{FM} 2009: Formal Methods, Second World Congress,
  Eindhoven, The Netherlands, November 2-6, 2009. Proceedings}}, {\sl
  \bibinfo{series}{Lecture Notes in Computer Science}} \bibinfo{volume}{5850},
  \bibinfo{publisher}{Springer}, pp. \bibinfo{pages}{1--15},
  \doi{10.1007/978-3-642-05089-3\_1}.

\bibitemdeclare{article}{DBLP:journals/iandc/Vasconcelos12}
\bibitem{DBLP:journals/iandc/Vasconcelos12}
\bibinfo{author}{Vasco~T. \surnamestart Vasconcelos\surnameend}
  (\bibinfo{year}{2012}): \emph{\bibinfo{title}{Fundamentals of session
  types}}.
\newblock {\sl \bibinfo{journal}{Inf. Comput.}} \bibinfo{volume}{217}, pp.
  \bibinfo{pages}{52--70}, \doi{10.1016/j.ic.2012.05.002}.

\bibitemdeclare{article}{vivas02}
\bibitem{vivas02}
\bibinfo{author}{Jos{\'{e}}{-}Luis \surnamestart Vivas\surnameend} \&
  \bibinfo{author}{Nobuko \surnamestart Yoshida\surnameend}
  (\bibinfo{year}{2002}): \emph{\bibinfo{title}{Dynamic Channel Screening in
  the Higher Order pi-Calculus}}.
\newblock {\sl \bibinfo{journal}{Electr. Notes Theor. Comput. Sci.}}
  \bibinfo{volume}{66}(\bibinfo{number}{3}), pp. \bibinfo{pages}{170--184},
  \doi{10.1016/S1571-0661(04)80421-3}.

\bibitemdeclare{article}{westin2003social}
\bibitem{westin2003social}
\bibinfo{author}{Alan~F \surnamestart Westin\surnameend}
  (\bibinfo{year}{2003}): \emph{\bibinfo{title}{Social and political dimensions
  of privacy}}.
\newblock {\sl \bibinfo{journal}{Journal of social issues}}
  \bibinfo{volume}{59}(\bibinfo{number}{2}), pp. \bibinfo{pages}{431--453},
  \doi{10.1111/1540-4560.00072}.

\end{thebibliography}

\appendix
\appendix

\renewcommand{\thesection}{A}


\section{Proofs from Section~\ref{sec:behavioral}}
\label{app:proof1}


 


\begin{prop}{\ref{prop:behavioural}}
For any process $\PP$, channel $\NM$ and prefix $\pi$, next equality holds

\[
\rest\NK
( (\rest\NL\send\NK\role\msg\NL.\receive\NM\role\msg\NY.\match\NY\role\msg\NL\pi.\inact) \parop \receive\NK\role\msg\NX.\PP) 
\sim 
\rest\NK
( (\rest\NL\send\NK\role\msg\NL.\receive\NM\role\msg\NY.\inact) \parop \receive\NK\role\msg\NX.\PP).
\]

\end{prop}
\begin{proof}
The proof follows by coinduction, by showing that the relation 
\[
\begin{array}[t]{@{}rcl@{\;\;}l@{}}
{\cal R} & =  & \{ \big(
\rest\NK
( \rest\NL\send\NK\role\msg\NL.\receive\NM\role\msg\NY.\match\NY\role\msg\NL\pi.\inact \parop \receive\NK\role\msg\NX.\PP), 
\rest\NK
( \rest\NL\send\NK\role\msg\NL.\receive\NM\role\msg\NY.\inact \parop \receive\NK\role\msg\NX.\PP)
\big),  \\
          &   & \;\; \big(
\rest\NK\rest\NL
( \receive\NM\role\msg\NY.\match\NY\role\msg\NL\pi.\inact \parop \PQ), 
\rest\NK \rest\NL
( \receive\NM\role\msg\NY.\inact \parop \PQ)
\big), \\
          &   & \;\; \big(
\rest\NL
( \receive\NM\role\msg\NY.\match\NY\role\msg\NL\pi.\inact \parop \PQ), 
 \rest\NL
( \receive\NM\role\msg\NY.\inact \parop \PQ)
\big), \\
          &   & \;\; \big(
\rest\NK \rest\NL
( \match\NN\role\msg\NL\pi.\inact \parop \PQ), 
\rest\NK \rest\NL
(\inact \parop \PQ)
\big), \\
          &   & \;\; \big(
\rest\NL
( \match\NN\role\msg\NL\pi.\inact \parop \PQ), 
 \rest\NL
 (\inact \parop \PQ)
\big)\\
          &   & \;\; \big(
\rest\NK \rest\NL \rest\NN
( \match\NN\role\msg\NL\pi.\inact \parop \PQ), 
\rest\NK \rest\NL \rest\NN
(\inact \parop \PQ)
\big), \\
          &   & \;\; \big(
\rest\NL \rest\NN
( \match\NN\role\msg\NL\pi.\inact \parop \PQ), 
 \rest\NL \rest\NN
 (\inact \parop \PQ)
\big) \\
&   & \;\;\;|\; \mbox{for all}\; \NN,\NM\in\Chn, \; \mbox{such that} \; \NN\not=\NL, 
\;\hbox{ and all processes}\; \PP \; \mbox{and} \;  \PQ, \; \hbox{such that} \; \NL\notin\fo\PQ\,\}
\end{array}
\]
where $\NN\not=\NL$, is contained in the strong bisimilarity, i.e., ${\cal R}\subseteq \sim$.

We show that each action of one process can be mimicked by the other process in the pair in $\cal R$, 
leading to 
processes that are again in relation $\cal R$. 
Let the process in the first pair
\[
\rest\NK
( \rest\NL\send\NK\role\msg\NL.\receive\NM\role\msg\NY.\match\NY\role\msg\NL\pi.\inact \parop \receive\NK\role\msg\NX.\PP)\lts{\alpha} \PP'.
\]
Then, since actions of the starting process can only be actions of its two branches,
we conclude that either $\alpha=\rest\NL\send\NK\role\msg\NL$ or $\alpha=\receive\NK\role\msg\NN$ or 
it is the synchronization of these two actions, in which case $\alpha=\tau$. 
We reject the first two options, since the subject of the action is bound in the starting process and by 
rule \rulename{(res)} it cannot be observed outside of the process.
Hence, we conclude $\alpha=\tau$ and $\PP'=\rest\NK \rest\NL(\receive\NM\role\msg\NY.\match\NY\role\msg\NL\pi.\inact \parop \PP\subst\NL\NX)$. 
Then, by applying \rulename{(out)}, \rulename{(open)}, \rulename{(in)}, \rulename{(close-l)} and \rulename{(res)}, respectively, we get 
\[
\rest\NK
( \rest\NL\send\NK\role\msg\NL.\receive\NM\role\msg\NY.\inact \parop \receive\NK\role\msg\NX.\PP)
\lts{\tau} 
\rest\NK \rest\NL(\receive\NM\role\msg\NY.\inact \parop \PP\subst\NL\NX),
\]
and since $\NL\notin\fn\PP$ and $\NX$ cannot appear as an object 
in the prefixes in $\PP$ we conclude $\NL\notin\fo{\PP\subst\NL\NX}$. 
Hence,  we have 
$\big(\rest\NK \rest\NL(\receive\NM\role\msg\NY.\match\NY\role\msg\NL\pi.\inact \parop \PP\subst\NL\NX),
\rest\NK \rest\NL(\receive\NM\role\msg\NY.\inact \parop \PP\subst\NL\NX)\big)\in{\cal R}$. 
The symmetric case is analogous.

Now let us consider processes in the second pair of $\cal R$. 
If 
\[
\rest\NK \rest\NL(\receive\NM\role\msg\NY.\match\NY\role\msg\NL\pi.\inact \parop \PQ)\lts{\alpha}\PP',
\]
then observable $\alpha$ can originate from both of the branches or from their synchronization.

\;\emph{---Left branch:}\; 
If the observable originate from the left branch, then $\alpha=\receive\NM\role\msg\NN$, and by 
\rulename{(in)}, \rulename{(par-l)} and \rulename{(res)} we get 
\[
\rest\NK \rest\NL(\receive\NM\role\msg\NY.\match\NY\role\msg\NL\pi.\inact \parop \PQ)
\lts{\receive\NM\role\msg\NN} 
\rest\NK \rest\NL(\match\NN\role\msg\NL\pi.\inact \parop \PQ),
\]
where, by the side condition of \rulename{(res)} we conclude $\NN\notin\{\NK, \NL\}$.
In the same way we get 
\[
\rest\NK \rest\NL(\receive\NM\role\msg\NY.\inact \parop \PQ)
\lts{\receive\NM\role\msg\NN}
\rest\NK \rest\NL(\inact \parop \PQ),
\]
and 
$\big(\rest\NK \rest\NL(\match\NN\role\msg\NL\pi.\inact \parop \PQ), 
\rest\NK \rest\NL(\inact \parop \PQ)\big)\in{\cal R}$ holds.

\;\emph{---Right branch:}\;
If the action originates from the right branch, i.e., 
from $\PQ \lts{\alpha} \PQ'$,
we distinguish two cases:
\begin{itemize}
\item [(i)] if by rules \rulename{(par-r)} and \rulename{(res)} is derived 
\[
\rest\NK \rest\NL(\receive\NM\role\msg\NY.\match\NY\role\msg\NL\pi.\inact \parop \PQ)
\lts{\alpha} 
\rest\NK \rest\NL(\receive\NM\role\msg\NY.\match\NY\role\msg\NL\pi.\inact \parop \PQ'),
\]
where we conclude that $\NK, \NL\notin\n\alpha$, hence $\NL\notin\fo{\PQ'}$.
Then by the same rules we get 
\[
\rest\NK \rest\NL(\receive\NM\role\msg\NY.\inact \parop \PQ)
\lts{\alpha} 
\rest\NK \rest\NL(\receive\NM\role\msg\NY.\inact \parop \PQ'),
\]
and 
$\big(\rest\NK \rest\NL(\receive\NM\role\msg\NY.\match\NN\role\msg\NL\pi.\inact \parop \PQ'),
\rest\NK \rest\NL(\receive\NM\role\msg\NY.\inact \parop \PQ')\big)\in{\cal R}$ holds.
\item [(ii)] if by rules \rulename{(par-r)}, \rulename{(res)} and \rulename{(open)} is derived 
\[
\rest\NK \rest\NL(\receive\NM\role\msg\NY.\match\NY\role\msg\NL\pi.\inact \parop \PQ)
\lts{\rest\NK\alpha} 
\rest\NL(\receive\NM\role\msg\NY.\match\NY\role\msg\NL\pi.\inact \parop \PQ'),
\]
then $\NL\notin\n\alpha$. 
Notice that the scope of channel $\NL$ cannot be 
extruded this way since 
$\NL\notin\fo\PQ$.
Hence, process $\PQ$ cannot 
perform output action with object $\NL$.
Then by the same rules we get 
\[
\rest\NK \rest\NL(\receive\NM\role\msg\NY.\inact \parop \PQ)
\lts{\rest\NK\alpha} 
 \rest\NL(\receive\NM\role\msg\NY.\inact \parop \PQ'),
\]
and, again, 
$\big( \rest\NL(\receive\NM\role\msg\NY.\match\NN\role\msg\NL\pi.\inact \parop \PQ'),
 \rest\NL(\receive\NM\role\msg\NY.\inact \parop \PQ')\big)\in{\cal R}$ holds.
\end{itemize} 
%

\;\emph{---Synchronization of branches:}\;
We again distinguish two cases:
\begin{itemize}
\item [(i)] if from 
\[
\receive\NM\role\msg\NY.\match\NY\role\msg\NL\pi.\inact\lts{\receive\NM\role\msg\NN}\match\NN\role\msg\NL\pi.\inact \qquad \mbox{and} \qquad 
\PQ\lts{\send\NM\role\msg\NN}\PQ',
\]
where we can make the same observation on $\PQ$ as before to conclude that $\NL\not=\NN$,  
by rules \rulename{(comm-r)} and \rulename{(res)} is derived
\[
\rest\NK \rest\NL(\receive\NM\role\msg\NY.\match\NY\role\msg\NL\pi.\inact \parop \PQ)
\lts{\tau}
\rest\NK \rest\NL(\match\NN\role\msg\NL\pi.\inact \parop \PQ').
\]
Then, using $\receive\NM\role\msg\NY.\inact\lts{\receive\NM\role\msg\NN}\inact$,
and the same rules as above we get
\[
\rest\NK \rest\NL(\receive\NM\role\msg\NY.\inact \parop \PQ)
\lts{\tau}
\rest\NK \rest\NL(\inact \parop \PQ'),
\]
and we get $\big( \rest\NK \rest\NL(\match\NN\role\msg\NL\pi.\inact \parop \PQ'),
\rest\NK \rest\NL(\inact \parop \PQ')\big)\in{\cal R}$.
\item [(ii)] 
if from 
\[
\receive\NM\role\msg\NY.\match\NY\role\msg\NL\pi.\inact\lts{\receive\NM\role\msg\NN}\match\NN\role\msg\NL\pi.\inact \qquad \mbox{and} \qquad 
\PQ\lts{\rest\NN\send\NM\role\msg\NN}\PQ',
\]
where 
as before we can assume $\NL\not=\NN$,  
by rules \rulename{(close-r)} and \rulename{(res)} is derived
\[
\rest\NK \rest\NL(\receive\NM\role\msg\NY.\match\NY\role\msg\NL\pi.\inact \parop \PQ)
\lts{\tau}
\rest\NK \rest\NL\rest\NN(\match\NN\role\msg\NL\pi.\inact \parop \PQ'),
\]
then using $\receive\NM\role\msg\NY.\inact\lts{\receive\NM\role\msg\NN}\inact$,
we may observe
\[
\rest\NK \rest\NL(\receive\NM\role\msg\NY.\inact \parop \PP)
\lts{\tau}
\rest\NK \rest\NL\rest\NN(\inact \parop \PQ'),
\]
and $\big( \rest\NK \rest\NL\rest\NN(\match\NN\role\msg\NL\pi.\inact \parop \PQ'),
\rest\NK \rest\NL\rest\NN(\inact \parop \PQ')\big)\in{\cal R}$.
\end{itemize}

The symmetric cases and the rest of the pairs from $\cal R$ are analogous. 
For the rest of the pairs note that in all left branch of the first components we have 
$\match\NN\role\msg\NL\pi.\inact$, where $\NN\not=\NL$, and hence, its observational power is 
equivalent to the observational power of inactive process $\inact$ appearing as the left branch in the right components.


\end{proof}

\begin{prop}{\ref{prop:non-forwarding-of-pi-processes}}
Let $\PP$ be a (sum-free) $\pi$-calculus process. If there is a $C_\pi$ process $\PQ$, such that $\PP\sim\PQ$, then $\PP$ satisfies the non-forwarding property.

\end{prop}

\begin{proof}
Let $\PP_1=\PP$ be a be a (sum-free) $\pi$-calculus process and let $\PP_1\lts{\alpha_1}\PP_2\lts{\alpha_2}\ldots\lts{\alpha_{m}}\PP_{m+1}.$
Let us fix $i\in\{1,\ldots, m-1\}$, and assume $\NL\notin\fn{\PP_i}$ and $\alpha_i=\receive\NK\role\msg\NL$. Since without loss of generality we can assume all bound outputs are fresh, we get $\alpha_j\not=\rest\NL\send{\NK'}\role\msg\NL$, for all $j=i+1, \ldots,m$, directly. 
In addition to the first assumption, let us assume there is $j\in\{i+1, \ldots,m\}$ such that $\alpha_j=\send{\NK'}\role\msg\NL$. 
Since $\PP_1\sim\PQ_1$ (where $\PQ_1=\PQ$), we conclude there are $C_\pi$ processes $\PQ_2, \ldots, \PQ_{m+1}$ such that
\[
\PQ_1\lts{\alpha_1}\PQ_2\lts{\alpha_2}\ldots\lts{\alpha_{m}}\PQ_{m+1}
\] 
and $\PP_n\sim\PQ_n$, for all $n=1,\ldots,m+1$, where $\PQ_i\lts{\receive\NK\role\msg\NL}\PQ_{i+1}$ and $\PQ_j\lts{\send{\NK'}\role\msg\NL}\PQ_{j+1}$. 
We now distinguish two cases.
\begin{enumerate}
\item If $\NL\notin\fn{\PQ_i}$ then we get a direct contradiction with Theorem~\ref{the:non-forwarding}. 
\item If $\NL\in\fn{\PQ_i}$, we choose a fresh channel $\NL'$ and a substitution $\sigma$ that is defined only on channel $\NL$ and maps it to $\NL'$. Then, from $\PP_j\lts{\alpha_j}\PP_{j+1}$, by consequitive application of Lemma~$1.4.8$ of~\cite{pi-calculus}, we conclude $(\PP_j)\sigma\lts{(\alpha_j)\sigma}(\PP_{j+1})\sigma$, for all $j=i+1, \ldots,m$. Since $\NL\notin\fn{\PP_i}$ we get $(\PP_i)\sigma=\PP_i$.
Now from 
\[
\PP_1\lts{\alpha_1}\ldots\lts{\alpha_{i-1}}\PP_i\lts{(\alpha_i)\sigma}(\PP_{i+1})\sigma\lts{(\alpha_{i+1})\sigma}\ldots\lts{(\alpha_{m})\sigma}(\PP_{m+1})\sigma,
\]
and $\PP_1\sim\PQ_1$, we again conclude there are $C_\pi$ processes $\PQ_2, \ldots, \PQ_{m+1}$ such that
\[
\PQ_1\lts{\alpha_1}\ldots\lts{\alpha_{i-1}}\PQ_i\lts{(\alpha_i)\sigma}\PQ_{i+1}\lts{(\alpha_{i+1})\sigma}\ldots\lts{(\alpha_{m})\sigma}\PQ_{m+1},
\]
where $\PP_j\sim\PQ_j$, for all $j=1,\ldots,i$ and $(\PP_j)\sigma\sim\PQ_j$, for all $j=i+1,\ldots,m+1$. 
Since $\NL'$ has been chosen to be a fresh channel, we get $\NL'\notin\fn{\PQ_i}$, and since  $\PQ_i\lts{\receive{(\NK)\sigma}\role\msg{\NL'}}\PQ_{i+1}$ and $\PQ_j\lts{\send{(\NK')\sigma}\role\msg{\NL'}}\PQ_{j+1}$, we fall into the first case, and hence, we again get contradiction with Theorem~\ref{the:non-forwarding}. 

\end{enumerate}
\end{proof}

\appendix
\renewcommand{\thesection}{B}
\section{Proofs from Section~\ref{sec:encoding}}\label{app:proof2}

\paragraph{Abbreviations.} For the sake of readability, use the abbreviation
\[
\handler_\NK=\rep\receive{\NN_\NK}\role\msg\NX.\send\NX\role\msg\NK .\inact
\parop 
\rep\receive{\NM_\NK}\role\msg{(\NX_1, \NX_2)}.\receive{\NX_1}\role\msg\NY. \rest\NT\send{\NY}\role\msg{(\NK, \NN_\NK, \NM_\NK, \NT)}.\send{\NX_2}\role\msg\NT.\inact,
\] 
assuming that $\varphi_{\enc{\;}}(\NK)=(\NK, \NN_\NK, \NM_\NK)$, 
and we omit writing trailing $\inact$'s whenever possible.


%
%

%
%
%
%
%
%
We may notice that the encoding defined in Table~\ref{tab:Encoding} does not introduce any free names by itself, except in the rule for output (and matching prefixed output), where names $\NN_\NA$ and $\NM_\NB$ are introduced. We may also notice also that these introduced names are the ones specified in the renaming policy of names $\NA$ and $\NB$, respectively. Hence, the following result is straightforward.

\begin{lemma}[Name invariance] Let $\PP$ be a $\pi$-calculus process and let substitutions $\sigma$ and $\sigma'$ be such that $\varphi_{\enc{\;}}((\NA)\sigma)=(\varphi_{\enc{\;}}(\NA))\sigma'$, for all $\NA\in\N$.
Then $\enc{(\PP)\sigma}=(\enc{\PP})\sigma'$.
\end{lemma}

For the operational correspondence, we will need only one case of the name invariance result, and we state it in the next corollary.

\begin{corollary}\label{lem:encode-subst}
If $\varphi_{\enc{\;}}(\NK)=(\NK, \NN_\NK, \NM_\NK)$ and $\varphi_{\enc{\;}}(\NX)=(\NX, \NN_\NX, \NM_\NX)$ then 
\[
\enc{\PP}\subst{\NK}{\NX}\subst{\NN_\NK}{\NN_\NX}\subst{\NM_\NK}{\NM_\NX}= \enc{\PP\subst{\NK}{\NX}}.
\]
\end{corollary}



For the operational correspondence we will need definition of strong bisimilarity of 
polyadic $C_\pi$-calculus. Such definition is exactly the same as Definition~\ref{def:strong-bisimilarity},
 except that labels of the LTS (where rules in Table~\ref{tab:Transition} are adapted for polyadic calculus following expected lines) are carrying a tuple of channels, i.e., 
\[
\alpha ::= \quad \send\NK\role\msg{(\NL_1, \ldots,\NL_n)} \quad \parop \quad \receive\NK\role\msg{(\NL_1, \ldots,\NL_n)} \quad \parop \quad \rest{\NL_1, \ldots,\NL_m}\send\NK\role\msg{(\NL_1, \ldots,\NL_n)} \quad \parop \quad \tau
\]
For such strong bisimilarity relation $\sim$ we may show to obey some standard properties stated in the next proposition.
\begin{proposition}\label{prop:bisim-standard-props}
\begin{enumerate}
\item $\sim$ is an equivalence relation and a non-input congruence;
\item $\match\NA\role\msg\NA\pi.\PP \sim \pi.\PP$;
\item $\PP_1 \parop (\PP_2 \parop \PP_3) \sim (\PP_1 \parop \PP_2) \parop \PP_3$;
\item $\PP_1 \parop \PP_2 \sim \PP_2 \parop \PP_1$;
\item $\PP \parop \inact \sim \PP$;
\item $\rest\NK\rest\NL\PP \sim \rest\NL\rest\NK\PP$;
\item $\rest\NK \inact \sim \inact$;
\item $\PP_1 \parop \rest\NK\PP_2 \sim \rest\NK (\PP_1 \parop \PP_2)$, if $\NK\notin\fn{\PP_1}$;
\item $\rep\PP \sim \PP \parop \rep\PP$;
\item $\rest\NK \rep\receive\NK\role\msg{(\NX_1, \ldots, \NX_n)}.\PP \sim \inact$;
\item $\rest{\NK, \NN_\NK, \NM_\NK} \handler_\NK \sim \inact$.
\end{enumerate}
\end{proposition}

We take the structural congruence relation $\equiv$ as it is defined  for the $\pi$-calculus processes in~\cite{pi-calculus}.

\begin{lemma}\label{lem:encode-struct-to-bisim}
If $\PP\equiv\PQ$ then $\enc{\PP} \sim \enc{\PQ}$.
\end{lemma}
\begin{proof}
The proof is by case analysis on the structural congruence rule applied. 
\begin{enumerate}
\item $\match\NA\role\msg\NA\pi.\PP \equiv \pi.\PP$.

We distinguish two cases for prefix $\pi$. 
	\begin{enumerate}
	\item If $\pi=\match{\NB_1}\role\msg{\NC_1}\ldots\match{\NB_n}\role\msg{\NC_n}\receive\ND\role\msg\NX$,
	then by definition of the encoding and Proposition~\ref{prop:bisim-standard-props} we get 
	\[
	\begin{array}{l@{\;}c@{\;}l}
	\enc{\match\NA\role\msg\NA\pi.\PP} 
	& = &
	\match\NA\role\msg\NA\match{\NB_1}\role\msg{\NC_1}\ldots\match{\NB_n}\role\msg{\NC_n}\receive\ND\role\msg{(\NX, \NN_\NX, \NM_\NX, \NX')}.\receive{\NX'}\role\msg\NY.\enc{\PP} \\
	&\sim &
	\match{\NB_1}\role\msg{\NC_1}\ldots\match{\NB_n}\role\msg{\NC_n}\receive\ND\role\msg{(\NX, \NN_\NX, \NM_\NX, \NX')}.\receive{\NX'}\role\msg\NY.\enc{\PP} \\
	& = &
	\enc{\match{\NB_1}\role\msg{\NC_1}\ldots\match{\NB_n}\role\msg{\NC_n}\receive\ND\role\msg\NX.\PP} \\
	& = & 
	\enc{\pi.\PP}.
	\end{array}
	\]

	\item If $\pi=\match{\NB_1}\role\msg{\NC_1}\ldots\match{\NB_n}\role\msg{\NC_n}\send\ND\role\msg\NGG$,
	then, again, by definition of the encoding and Proposition~\ref{prop:bisim-standard-props} we get 

	\[
	\begin{array}{l@{\;}c@{\;}l}
	\enc{\match\NA\role\msg\NA\pi.\PP} 
	& = &  \rest{\NE_1, \NE_2}\match\NA\role\msg\NA\match{\NB_1}\role\msg{\NC_1}\ldots\match{\NB_n}\role\msg{\NC_n}\send{\NN_\ND}\role\msg{\NE_1}.\send{\NM_\NGG}\role\msg{(\NE_1, \NE_2)}.\receive{\NE_2}\role\msg\NX.\send\NX\role\msg{\NE_1}.\enc{\PP}\\
	& \sim &  \rest{\NE_1, \NE_2}\match{\NB_1}\role\msg{\NC_1}\ldots\match{\NB_n}\role\msg{\NC_n}\send{\NN_\ND}\role\msg{\NE_1}.\send{\NM_\NGG}\role\msg{(\NE_1, \NE_2)}.\receive{\NE_2}\role\msg\NX.\send\NX\role\msg{\NE_1}.\enc{\PP}\\
	& = &
	\enc{\match{\NB_1}\role\msg{\NC_1}\ldots\match{\NB_n}\role\msg{\NC_n}\send\ND\role\msg\NE.\PP} \\
	& = & 
	\enc{\pi.\PP}.
	\end{array}
	\]
	\end{enumerate}

\item $\rest\NK\rest\NL\PP \equiv \rest\NL\rest\NK\PP$. By the definition of the encoding and Proposition~\ref{prop:bisim-standard-props} we get 
 
\[
\begin{array}{l@{\,}c@{\,}l}
\enc{\rest\NK\rest\NL\PP}& = & 
\rest{\NK, \NN_\NK, \NM_\NK}( \rest{\NL, \NN_\NL, \NM_\NL}(\enc{\PP}\parop \handler_\NL )\parop \handler_\NK)\\
& \sim &  \rest{\NL, \NN_\NL, \NM_\NL}(\rest{\NK, \NN_\NK, \NM_\NK} (\enc{\PP} \parop \handler_\NK) \parop \handler_\NL )) \\
& = & \enc{\rest\NL\rest\NK\PP}.
\end{array}
\]

\item $\rest\NK\inact \equiv \inact$. By the definition of the encoding and Proposition~\ref{prop:bisim-standard-props} we get 
\[
\begin{array}{l@{\,}c@{\,}l}
\enc{\rest\NK\inact} & = &
\rest{\NK, \NN_\NK, \NM_\NK}(\inact \parop \handler_\NK) \\
& \sim & \rest{\NK, \NN_\NK, \NM_\NK}\handler_\NK \\
&\sim & \inact= \enc{\inact}.\\
\end{array}
\]

\item $\PP \parop \rest\NA\PQ \equiv \rest\NA (\PP \parop \PQ)$, if $\NA\notin\fn\PP$. By the definition of the encoding and Proposition~\ref{prop:bisim-standard-props} we get 
\[
\begin{array}{l@{\,}c@{\,}l}
\enc{\PP \parop \rest\NK\PQ} & = & 
\enc{\PP} \parop \rest{\NK, \NN_\NK, \NM_\NK} (\enc{\PQ} \parop \handler_\NK) \\
& \sim & 
\rest{\NK, \NN_\NK, \NM_\NK}(\enc{\PP} \parop  \enc{\PQ} \parop \handler_\NK) \\
& = & 
\rest{\NK, \NN_\NK, \NM_\NK}(\enc{\PP \parop  \PQ} \parop \handler_\NK) \\
& = & 
\enc{\rest\NL(\PP \parop \PQ)}.
\end{array}
\]

\item The rest of the cases are analogous.
\end{enumerate}
\end{proof}
\begin{lemma}\label{fnn}
Let $\PP$ and $\PQ$ be $\pi$-calculus processes. 
\begin{enumerate}
\item If $\PP\equiv\PQ$ then $\fnn{\PP}=\fnn{\PQ}$.
\item If $\NA\notin\fnn{\PP}$ then there exist a $\pi$-calculus process $\PP'$ such that $\PP\equiv\PP'$ and $\NA\notin\fn{\PP'}$.
\item if $\NA\notin\fnn\PP$ and $\PP\red\PQ$ then $\NA\notin\fnn\PQ$.
\end{enumerate} 
\end{lemma}
\begin{proof}
\begin{enumerate}
\item The only structural congruence rule affecting free names is $\match\NA\role\msg\NA\pi.\PP\equiv\pi.\PP$, and by the definition $\fnn{\match\NA\role\msg\NA\pi.\PP}=\fnn{\pi.\PP}$.
\item Assume $\NA\notin\fnn{\PP}$. If $\NA\notin\fn{\PP}$ then the proof is finished. 
Now assume $\NA\in\fn{\PP}$. Then $\NA$ can appear only in $\PP$ in a sub-process of the form $\match\NA\role\msg\NA\pi.\PQ$. In this case we may show, by induction on the structure of $\PP$, that using the structural congruence rule $\match\NA\role\msg\NA\pi.\PQ\equiv\pi.\PQ$, we can get rid of all such matchings that mention name $\NA$.
\item Follows by an easy induction on $\red$ derivation.
\end{enumerate}
\end{proof}
\begin{lemmata}{\ref{lem:oper-corresp-with-Hs}}
If $\PP \red \PQ$ then 
$
\enc{\PP}\parop \handler \ltss{\tau}* \sim \enc{\PQ}\parop \handler,
$
where
\begin{itemize}
\item if $\fnn{\PP}=\{\NK_1, \ldots, \NK_n\}$, then 
\[
\handler=\prod\limits_{i\in\{1, \ldots, n\}} \handler_{\NK_i},
\]
\item if $\fnn{\PP}=\emptyset$ then 
$\handler=\inact$.
\end{itemize}
\end{lemmata}
\begin{proof}
The proof is by induction on $\red$ derivation.
\begin{enumerate}
\item \emph{Base case}: $\send\NK\role\msg\NL.\PP \parop \receive\NK\role\msg\NX.\PQ \red \PP \parop \PQ\subst{\NL}{\NX}$. Since $\NK$ and $\NL$ are free in the starting process, we can encode it as 
\[
\PR=\enc{\send\NK\role\msg\NL.\PP \parop \receive\NK\role\msg\NX.\PQ} \parop \handler_\NK \parop \handler_\NL \parop \handler,
\]
where 
if 
$\fnn{\send\NK\role\msg\NL.\PP \parop \receive\NK\role\msg\NX.\PQ}=\{\NK, \NL, \NK_1, \ldots, \NK_n\}$ 
then 
$\handler=\prod\limits_{i\in\{1, \ldots, n\}} \handler_{\NK_i}.$ 
If $\fnn{\send\NK\role\msg\NL.\PP \parop \receive\NK\role\msg\NX.\PQ}=\{\NK, \NL\}$ then $\handler=\inact$.
Then, 
\[
\begin{array}{l@{\,}c@{\,}l}
\PR 
& = & 
\enc{\send\NK\role\msg\NL.\PP} \parop \enc{\receive\NK\role\msg\NX.\PQ} \parop \handler_\NK \parop \handler_\NL \parop \handler\\
& = & 
\rest{\NE_1, \NE_2} \send{\NN_\NK}\role\msg{\NE_1}.\send{\NM_\NL}\role\msg{(\NE_1, \NE_2)}.\receive{\NE_2}\role\msg\NY.\send\NY\role\msg{\NE_1}.\enc{\PP} \parop \receive\NK\role\msg{(\NX, \NN_\NX, \NM_\NX, \NX')}.\receive{\NX'}\role\msg\NY.\enc{\PQ} \\
& & \parop \rep\receive{\NN_\NK}\role\msg\NY.\send\NY\role\msg\NK \parop \rep\receive{\NM_\NK}\role\msg{(\NX_1, \NX_2)}.\receive{\NX_1}\role\msg\NY. \rest\NT\send{\NY}\role\msg{(\NK, \NN_\NK, \NM_\NK, \NT)}.\send{\NX_2}\role\msg\NT  \\
& & \parop \rep\receive{\NN_\NL}\role\msg\NY.\send\NY\role\msg\NL \parop \rep\receive{\NM_\NL}\role\msg{(\NX_1, \NX_2)}.\receive{\NX_1}\role\msg\NY. \rest\NT\send{\NY}\role\msg{(\NL, \NN_\NL, \NM_\NL, \NT)}.\send{\NX_2}\role\msg\NT \parop \handler  \\
& \lts{\tau}\lts{\tau} & 
\rest{\NE_1,\NE_2} (\receive{\NE_2}\role\msg\NY.\send\NY\role\msg{\NE_1}.\enc{\PP} \parop \receive\NK\role\msg{(\NX, \NN_\NX, \NM_\NX, \NX')}.\receive{\NX'}\role\msg\NY.\enc{\PQ} \\
& & \parop \send{\NE_1}\role\msg\NK
\parop \handler_\NK  \\
& & \parop \rep\receive{\NN_\NL}\role\msg\NY.\send\NY\role\msg\NL 
\parop \receive{\NE_1}\role\msg\NY. \rest{\NT'}\send{\NY}\role\msg{(\NL, \NN_\NL, \NM_\NL, {\NT'})}.\send{\NE_2}\role\msg{\NT'}) \parop \rep\receive{\NM_\NL}\role\msg{(\NX_1, \NX_2)}.\receive{\NX_1}\role\msg\NY. \rest\NT\send{\NY}\role\msg{(\NL, \NN_\NL, \NM_\NL, \NT)}.\send{\NX_2}\role\msg\NT \parop \handler,  \\
\end{array}
\]
where the output process synchronize with the left thread of the handler of name $\NK$ and 
with the right thread of the handler of name $\NL$. 
At this point, the two handlers can synchronize and the last process evolves to
\[
\begin{array}{l@{\,}c@{\,}l}
& \lts{\tau}\lts{\tau}\lts{\tau} & 
\rest{\NE_2, \NE_1, \NT'}(\send{\NT'}\role\msg{\NE_1}.\enc{\PP} 
\parop \receive{\NT'}\role\msg\NY.\enc{\PQ}\subst{\NL}{\NX}\subst{\NN_\NL}{\NN_\NX}\subst{\NM_\NL}{\NM_\NX} \\
& & \parop \inact 
\parop \handler_\NK \\
& & \parop \rep\receive{\NN_\NL}\role\msg\NY.\send\NY\role\msg\NL 
\parop  \inact) \parop \rep\receive{\NM_\NL}\role\msg{(\NX_1, \NX_2)}.\receive{\NX_1}\role\msg\NY. \rest\NT\send{\NY}\role\msg{(\NL, \NN_\NL, \NM_\NL, \NT)}.\send{\NX_2}\role\msg\NT \parop \handler,  \\
\end{array}
\]
where, after the synchronization of the two handlers, name $\NL$ (together with $\NN_\NL$, $\NM_\NL$ and $\NT'$) 
is finally received in the input process, after which channel $\NT'$ is also received in the left-hand side process, 
making the encoding of  processes $\PP$ an $\PQ$ only unlocked in the synchronization:
\begin{equation}\label{eq-1}
\begin{array}{l@{\,}c@{\,}l}
& \lts{\tau} & 
\rest{\NE_2, \NE_1, \NT'}(\enc{\PP} 
\parop \enc{\PQ}\subst{\NL}{\NX}\subst{\NN_\NL}{\NN_\NX}\subst{\NM_\NL}{\NM_\NX} \\
& & \parop \inact 
\parop \handler_\NK  \\
& & \parop \rep\receive{\NN_\NL}\role\msg\NY.\send\NY\role\msg\NL 
\parop  \inact) \parop \rep\receive{\NM_\NL}\role\msg{(\NX_1, \NX_2)}.\receive{\NX_1}\role\msg\NY. \rest\NT\send{\NY}\role\msg{(\NL, \NN_\NL, \NM_\NL, \NT)}.\send{\NX_2}\role\msg\NT \parop \handler.  \\
\end{array}
\end{equation}
By Corollary~\ref{lem:encode-subst} we get 
\[
\enc{\PQ}\subst{\NL}{\NX}\subst{\NN_\NL}{\NN_\NX}\subst{\NM_\NL}{\NM_\NX}=
\enc{\PQ\subst{\NL}{\NX}}.
\] 
Hence, we conclude the last derived process in equation~(\ref{eq-1}) is equal to 
\[
\begin{array}{l@{\,}c@{\,}l}
&  & 
\rest{\NE_2, \NE_1, \NT'}(\enc{\PP} 
\parop \enc{\PQ\subst{\NL}{\NX}} \\
& & \parop \inact 
\parop \handler_\NK  \\
& & \parop \rep\receive{\NN_\NL}\role\msg\NY.\send\NY\role\msg\NL 
\parop  \inact) \parop \rep\receive{\NM_\NL}\role\msg{(\NX_1, \NX_2)}.\receive{\NX_1}\role\msg\NY. \rest\NT\send{\NY}\role\msg{(\NL, \NN_\NL, \NM_\NL, \NT)}.\send{\NX_2}\role\msg\NT \parop \handler.  \\
\end{array}
\]
Since $\NE_2, \NE_1,\NT' \notin \fn{\enc{\PP} 
\parop \enc{\PQ\subst{\NL}{\NX}}
\parop \inact \parop \handler_\NK
\parop \rep\receive{\NN_\NL}\role\msg\NY.\send\NY\role\msg\NL 
\parop  \inact}$, by Proposition~\ref{prop:bisim-standard-props} we have that 
the last derived process is strongly bisimilar to
\[
\begin{array}{l@{\,}c@{\,}l}
&  & 
\enc{\PP} 
\parop \enc{\PQ\subst{\NL}{\NX}}
\parop \rest{\NE_2} \rest{\NE_1} \rest{\NT}\inact 
\parop \handler_\NK 
\parop \handler_\NL 
\parop 
\handler  \\
& \sim & 
\enc{\PP} 
\parop \enc{\PQ\subst{\NL}{\NX}}
\parop \handler_\NK 
\parop \handler_\NL 
\parop 
\handler  \\
& = & 
\enc{\PP \parop \PQ\subst{\NL}{\NX}}
\parop \handler_\NK
\parop \handler_\NL \parop \handler.  \\
\end{array}
\]

\item  $\PP \parop \PR \red \PQ \parop \PR$ is derived from $\PP\red\PQ$. 
By induction hypothesis  
 
\[
\enc{\PP}\parop \handler_1 \ltss{\tau}* \PS,
\] 
where $\PS \sim \enc{\PQ}\parop \handler_1$ and  
if $\fnn{\PP}=\{\NK_1, \ldots, \NK_n\}$ then 

\[
\handler_1=\prod\limits_{i\in\{1, \ldots, n\}} \handler_{\NK_i},
\] 
and if $\fnn{\PP}=\emptyset$ then $\handler_1=\inact$.
Now, if $\fnn{\PR}\setminus\fnn{\PP}=\{\NL_1, \ldots, \NL_m\}$, 
let us take 
\[
\handler_2=\prod\limits_{j\in\{1, \ldots, m\}} \handler_{\NL_j}.
\]
If $\fnn{\PR}\setminus\fnn{\PP}=\emptyset$ let us take $\handler_2=\inact$.
%
Then, by \rulename{(par-l)} we can derive 
\[
\enc{\PP}\parop \handler_1 \parop \enc{\PR} \parop \handler_2 
\ltss{\tau}* 
\PS  \parop \enc{\PR}  \parop \handler_2.
\]

By Lemma~\ref{prop:bisim-standard-props} we get 
$\enc{\PP}\parop \handler_1 \parop \enc{\PR}\parop \handler_2 
\sim 
\enc{\PP} \parop \enc{\PR} \parop \handler_1 \parop \handler_2$
then 

\[
\enc{\PP \parop \PR} \parop \handler_1 \parop \handler_2 = \enc{\PP} \parop \enc{\PR} \parop \handler_1 \parop \handler_2 
\lts{\tau}^* \PS',
\]
where $\PS'  \sim  \PS  \parop \enc{\PR}  \parop \handler_2$, by the definition of strong bisimilarity. 
We can now conclude 
\[
\begin{array}{l@{\;}c@{\;}l}
\PS' & \sim & \PS  \parop \enc{\PR}  \parop \handler_2\\
      & \sim & \enc{\PQ}\parop \handler_1 \parop \enc{\PR}  \parop \handler_2\\
      & \sim & \enc{\PQ}  \parop \enc{\PR} \parop \handler_1  \parop \handler_2 \\
    & = & \enc{\PQ \parop \PR} \parop \handler_1  \parop \handler_2.\\
\end{array} 
\]

\item  $\rest\NK \PP \red \rest\NK \PQ$ is derived from $\PP\red\PQ$. Again, by induction hypothesis 
\[
\enc{\PP}\parop \handler_1 \ltss{\tau}* \PS,
\] 
where $\PS \sim \enc{\PQ}\parop \handler_1$ and if $\fnn{\PP}=\{\NK_1, \ldots, \NK_n\}$ then

\[
\handler_1=\prod\limits_{i\in\{1, \ldots, n\}} \handler_{\NK_i},
\] 
while if $\fnn\PP=\emptyset$ then $\handler_1=\inact$.
Since, 
\[
\enc{\rest\NK \PP} \parop \handler=\rest{\NK, \NN_\NK, \NM_\NK} (\enc{\PP} \parop \handler_{\NK})\parop \handler,
\] 
we distinguish two cases:
	\begin{enumerate}
	\item if $\NK\in\fnn{\PP}$ then 
	$\handler_\NK\parop \handler=\handler_1$.
	Since $\NK, \NN_\NK, \NM_\NK \notin\fn{\handler}$, by Proposition~\ref{prop:bisim-standard-props} we get
	\[
	\begin{array}{lcl}
	\rest{\NK, \NN_\NK, \NM_\NK} (\enc{\PP} \parop \handler_\NK)\parop \handler
	& \sim & 
	\rest{\NK, \NN_\NK, \NM_\NK} (\enc{\PP}  \parop \handler_\NK \parop \handler)\\
	& \ltss{\tau}*  & 
	\rest{\NK, \NN_\NK, \NM_\NK} \PS,\\
	\end{array}
	\]
	where $\ltss{\tau}*$ transition(s) follows by the induction hypothesis and rule \rulename{(res)}. 
	By Proposition~\ref{prop:bisim-standard-props} 
	\[
	\begin{array}{lcl}
	\rest{\NK, \NN_\NK, \NM_\NK} \PS 
	& \sim & 
	\rest{\NK, \NN_\NK, \NM_\NK} (\enc{\PQ}\parop \handler_1  )\\
	& = & 
	\rest{\NK, \NN_\NK, \NM_\NK} (\enc{\PQ}\parop \handler_\NK \parop \handler  )\\
	& \sim &
	\rest{\NK, \NN_\NK, \NM_\NK} (\enc{\PQ}\parop \handler_\NK)\parop \handler  \\
	& = & 
	\enc{\rest\NK\PQ}\parop \handler,
	\end{array} 
	\]
	and by definition and transitivity of strong bisimilarity we get 
	\[
	\enc{\rest\NK \PP} \parop \handler \ltss{\tau}* 
	\sim \enc{\rest\NK\PQ}\parop \handler.
	\]

	\item if $\NK\notin\fnn{\PP}$, then 
	$\handler= \handler_1$. By Lemma~\ref{fnn} there exist $\PP'$ such that $\PP\equiv\PP'$ and $\NK\notin\fn{\PP'}$. 
	Since $\equiv$ is a congruence and by Lemma~\ref{lem:encode-struct-to-bisim} we get $\enc{\PP}\sim\enc{ \PP'}$ and $\enc{\rest\NK \PP}\sim\enc{\rest\NK \PP'}$. Then, by definition of the encoding and Proposition~\ref{prop:bisim-standard-props} we have 
	\[
	\begin{array}{lcl} 
	\enc{\rest\NK \PP} \parop \handler_1
	& \sim & 
	\enc{\rest\NK \PP'} \parop \handler_1\\
	& = &
	\rest{\NK, \NN_\NK, \NM_\NK} (\enc{\PP'} \parop \handler_\NK)\parop \handler_1\\
	& \sim & 
	\enc{\PP'} \parop \rest{\NK, \NN_\NK, \NM_\NK} \handler_\NK \parop \handler_1\\
	& \sim & 
	\enc{\PP'} \parop \handler_1 \sim \enc{\PP} \parop \handler_1.\\
	\end{array}
	\]
	Since $\PP\red \PQ$ and $\NK\notin\fnn{\PP}$, by Lemma~\ref{fnn} we get $\NK\notin\fnn{\PQ}$.  By the same lemma we conclude there exist $\PQ'$ such that $\PQ\equiv\PQ'$ and $\NK\notin\fn{\PQ'}$. 
	Hence, again 
	\[
	\begin{array}{lcl}
	\enc{\PQ}\parop \handler_1
	& \sim &
	\enc{\PQ'}\parop \handler_1\\
	& \sim & \enc{\PQ'}\parop \rest{\NK, \NN_\NK, \NM_\NK}\handler_\NK \parop\handler_1 \\
	& \sim & \rest{\NK, \NN_\NK, \NM_\NK}(\enc{\PQ'}
	\parop \handler_\NK) \parop\handler_1 \\ 
	& = & \enc{\rest\NK\PQ'} \parop \handler_1 \sim \enc{\rest\NK\PQ} \parop \handler_1. \\  
	\end{array}
	\]
	By definition and transitivity of strong bisimilarity we get 
	\[
	\enc{\rest\NK \PP} \parop \handler_1 \ltss{\tau}* 
	\sim \enc{\rest\NK\PQ}\parop \handler_1.
	\]
	\end{enumerate}
\item $\PP'\red\PQ'$ is derived from $\PP\red\PQ$, where $\PP\equiv\PP'$ and $\PQ\equiv\PQ'$. 
By induction hypothesis
\[
\enc{\PP}\parop \handler_1 \ltss{\tau}* \PS,
\] 
where $\PS \sim \enc{\PQ}\parop \handler_1$, and if $\fnn{\PP}=\{\NK_1, \ldots, \NK_n\}$ then
\[
\handler_1=\prod\limits_{i\in\{1, \ldots, n\}} \handler_{\NK_i},
\] 
while if $\fnn\PP=\emptyset$ then $\handler_1=\inact$.
By Lemma~\ref{lem:encode-struct-to-bisim} and Lemma~\ref{fnn}, $\PP\equiv\PP'$ implies 
$\enc{\PP}\sim\enc{\PP'}$ and $\fnn{\PP}=\fnn{\PP'$}, and $\PQ\equiv\PQ'$ implies $\enc{\PQ}\sim\enc{\PQ'}$. 
Then, by Proposition~\ref{prop:bisim-standard-props} we get 
$\enc{\PP}\parop\handler_1 \sim \enc{\PP'}\parop \handler_1$, hence, 
by definition of strong bisimilarity 
\[
\enc{\PP'}\parop \handler_1 \ltss{\tau}* \PS',
\] 
where 
$\PS' \sim \PS \sim \enc{\PQ}\parop \handler_1 \sim \enc{\PQ'}\parop \handler_1$, 
which completes the proof.
\end{enumerate}
\end{proof}

As a direct consequence of Lemma~\ref{lem:oper-corresp-with-Hs}, we get the operational correspondence result for the encoding of closed $\pi$-calculus processes.

\begin{theorema}{\ref{theorem:operational-corresp}}
Let $\PP$ be a closed (sum-free) $\pi$-calculus process. 
If $\PP\red\PQ$ then $\enc{\PP}\ltss{\tau}* \sim\enc{\PQ}$.
\end{theorema}

\end{document}